%% file: main-iudx-t-it-v2.tex
\documentclass[journal,draftcls,onecolumn,12pt,twoside]{IEEEtran}
\usepackage{cite}
\usepackage{amsmath,amssymb,amsfonts,amsthm}
\usepackage{ifthen}
\usepackage{tikz}
\usepackage{cite}
\usepackage{url}
\usepackage{titlesec}
\usepackage{subfig}
\usepackage{blkarray}
\usepackage{dsfont}
\usepackage[mathscr]{euscript}
\usepackage{enumitem}
\usepackage{algorithm}
\usepackage[noend]{algpseudocode}
\usepackage{blkarray}
\newtheorem{theorem}{Theorem}[section]
\newtheorem{corollary}{Corollary}[section]
\newtheorem{proposition}{Proposition}[section]
\newtheorem{lemma}{Lemma}[section]
\theoremstyle{remark}

\theoremstyle{definition}
\newtheorem{definition}{Definition}[section]

\usepackage{subcaption}
\usepackage{blkarray}
\usepackage{dsfont}

\usepackage{url}
\usepackage{tcolorbox}
\usepackage{arydshln}
\usepackage{mathtools}
\usepackage{dsfont}

\usepackage{xcolor}

\usepackage{xcolor}

\definecolor{brickred}{cmyk}{0,0.89,0.94,0.28}
\definecolor{goldenrod}{cmyk}{0,0.10,0.84,0}
\definecolor{purple}{cmyk}{0.45,0.86,0,0}
\definecolor{rawsienna}{cmyk}{0,0.72,1,0.45}
\definecolor{olivegreen}{cmyk}{0.64,0,0.95,0.40}
\definecolor{peach}{cmyk}{0,0.5,0.7,0}
\definecolor{darkolive}{rgb}{0.,0.4,0.}
\colorlet{grey}{gray!40}

\global\long\def\P{\mathbb{P}}
\global\long\def\E{\mathbb{E}}

\global\long\def\d{\mathrm{d}}

\DeclareMathOperator*{\argmax}{arg\,max} \DeclareMathOperator*{\argmin}{arg\,min}

\usepackage[cmintegrals]{newtxmath}
\newlist{mylist}{enumerate}{3}
\setlist[mylist]{label={}}
\usepackage{graphicx}
\usepackage{textcomp}
\def\BibTeX{{\rm B\kern-.05em{\sc i\kern-.025em b}\kern-.08em
		T\kern-.1667em\lower.7ex\hbox{E}\kern-.125emX}}

\ifCLASSINFOpdf
\else
\fi

\begin{document}

\title{{On Improving the Composition Privacy Loss in Differential Privacy for Fixed Estimation Error}}


\author{
	\IEEEauthorblockN{V.~Arvind~Rameshwar,~\IEEEmembership{Member,~IEEE,}  \and
		Anshoo Tandon,~\IEEEmembership{Senior Member,~IEEE}}
	
	%
	\thanks{The authors {are with} the India Urban Data Exchange Program Unit, Indian Institute of Science, Bengaluru 560012, India, email: \texttt{arvind.rameshwar@gmail.com, anshoo.tandon@gmail.com}.
	}
}

\markboth{}%
{Rameshwar and Tandon: Improving the Composition Privacy Loss for Fixed Error}

%






\maketitle

\begin{abstract}
	This paper considers the private release of statistics of disjoint subsets of a dataset, {in the setting of data heterogeneity, where users could contribute more than one sample, with different users contributing potentially different numbers of samples}. In particular, we {focus on} the $\epsilon$-differentially private release of sample means and variances of sample values in disjoint subsets of a dataset, {under the assumption that the numbers of contributions of each user in each subset is publicly known}. 
	Our main contribution is an iterative algorithm, based on suppressing user contributions, which seeks to reduce the overall privacy loss degradation under a canonical Laplace mechanism, while not increasing the worst estimation error among the subsets. Important components of this analysis are our exact, analytical characterizations of the sensitivities and the worst-case bias errors of estimators of the sample mean and variance, which are obtained by clipping or suppressing user contributions. We test the performance of our algorithm on real-world and synthetic datasets and demonstrate clear improvements in the privacy loss degradation, for fixed {worst-case} estimation error.  
\end{abstract}


\section{Introduction}
\label{sec:intro}
Several landmark works have demonstrated that queries about seemingly benign functions of a dataset that is not publicly available can compromise the identities of the individuals in the dataset (see, e.g., \cite{sweeney,narayanan,homer,dinurnissim}). Examples of such reconstruction attacks for the specific setting of traffic datasets -- {a use-case that the results of this paper apply to} -- can be found in \cite{taxi, pandurangan}. In this context, the framework of (item-level) differential privacy (DP) was introduced in \cite{dwork06} (see also \cite{dworkroth}), which aims to preserve the privacy of users when each user contributes at most one sample, even in the presence of additional side information. 

The work {\cite{continual}} then considered the setting where users could contribute more than one sample and formalized the framework of \emph{user-level} DP, which requires the statistical indistinguishability of the output generated by a private mechanism, where potentially all of a user's contributions could be altered, from the output of the mechanism on the original dataset.  {{While early papers (see, e.g., \cite{userlevel}) on user-level DP} considered the setting where each user contributes the \emph{same} number of samples to the dataset, the work in \cite{hetero} extended the mean estimation algorithms \cite{userlevel} to settings where the dataset is ``heterogeneous'', in that different users could contribute different numbers of samples, from potentially different distributions.} {Later works explored PAC learning \cite{usereg1}, bounding user contributions in ML models \cite{usereg2}, and federated learning \cite{usereg3,usereg4, usereg5}, under user-level DP.} The focus on user-level privacy assumes significance in the context of most real-world IoT datasets, such as traffic databases, which record multiple contributions from every user, with different users contributing potentially different number of samples. 

Our interest in this work is in the private release of sample means and variances in several \emph{disjoint} subsets of a dataset, possibly sequentially, {in the setting where the samples are known to lie in a bounded interval}. Importantly, we assume that the partition of the dataset into subsets and the  ``occupancy array," or the numbers of contributions of each user in each subset, are publicly known, and fixed. Such partitioning of a dataset into ``grids" or subsets is natural in the setting of public transport datasets, datasets of Air Quality Indices (AQI), databases of call records, and datasets of credit card spends, where a given dataset is partitioned into disjoint grids \cite{h3} corresponding to disjoint areas in a city or country. We add that there exists work \cite{qardaji} on choosing partitions or grids in a differentially private manner. In what follows, we use the terminology ``grids'' to mean these disjoint subsets under consideration.

{The assumption that the occupancy array is public knowledge holds in the setting of, say public transport datasets, wherein the routes taken by the vehicles are fixed and known in advance, while the speeds (or other sensor readings) are private. Such an assumption is also natural in, say, repositories of crowd-funded investments, where the projects that are being funded (with each project considered as a subset or part of a ``grid'') are known, and so is the information about the number of times a particular user invested in a particular project over a fixed time span, while the actual amounts being invested are kept private. Another example where the occupancy array can be assumed to be known is the case where the number of visits of different individuals in different neighborhood stores, over a period of one month, for their daily needs are known, while the actual amount expended by an individual during a visit to a particular store is kept private. In this work, we hence focus on the setting of user-level DP mechanisms, which guarantee the statistical indistinguishability of their outputs to changes in a single user's contributions, while keeping the occupancy array fixed (see Section \ref{sec:dp} for details). This is a natural generalization of the ``replace" (or ``swap") model of DP (see, e.g., \cite{dwork06,vadhan}) to the setting of a partitioned dataset.}

{Further, the assumption that each user sample lies within a bounded interval is well-founded in the case of traffic datasets containing (non-negative) speed samples measured by sensors with a known maximum reading. Such an assumption is also natural in the setting of repositories of crowd-funded investments, wherein the investment amounts per issuer in a given time window are bounded, for deterring fraud and unethical non-disclosure \cite{crowd}, \cite[Sec. 302]{jobs}.}

Now, in the traditional setting of (pure) $\epsilon$-item-level DP or $\epsilon$-user-level DP (where $\epsilon$ captures the privacy loss), the {Basic Composition Theorem} shows that if the {analyst (or client)} were to pose multiple queries to the data curator in a potentially sequential (or adaptive) manner, the total privacy loss degrades by a factor that, in the worst case, equals the number of queries (see \cite[Cor. 3.15]{dworkroth})\footnote{{In fact, as we argue in Theorem \ref{thm:improvcomp}, this na\"ive bound on privacy loss degradation can be easily improved for the setting where the queries pertain to disjoint subsets of the dataset (see, e.g., \cite[Thm. 4]{pinq}).}}. It is also well-known that there exists a differentially private mechanism, namely, the canonical Laplace mechanism, which achieves this privacy loss (see, e.g., \cite[Sec. 2]{steinkechapter}). We mention that in the setting where we allow for (approximate) $(\epsilon,\delta)$-DP, for a certain range of $(\epsilon,\delta)$ values, it is possible to obtain improvements in the worst-case privacy loss as compared to that guaranteed by basic composition \cite[Sec. 3.5]{dworkroth}, \cite{boosting_composition, kairouz}. 
\subsection{Comparison with related work and contributions}
{
While traditional analysis of DP composition results \cite{dworkroth, boosting_composition, kairouz} focus exclusively on the privacy loss degradation assuming a generic differentially private estimator, a \emph{joint analysis} of the composition privacy loss and estimation error has been largely lacking. Often, in practice, a client who poses queries to a dataset and receives private answers has a fixed, but non-zero, error tolerance. Since the client is unaware of the true sample values in the dataset, he/she is willing to tolerate this fixed error for \emph{any} dataset queried -- in particular, for the \emph{worst} dataset in terms of error.}
 Our treatment, hence, is a study of composition of user-level DP mechanisms that \emph{jointly} considers the errors due to noise addition for privacy and due to the \emph{worst-case bias} (over all datasets) that results from the estimator used in the DP mechanism being different from the true function to be released. {We emphasize that this worst-case error (that is the sum of the worst-case bias error and a term proportional to the standard deviation of noise added) depends only on the publicly known occupancy array and interval in which the samples are known to lie, and is independent of the sample values themselves.} We mention that in recent work \cite{dp_preprint}, we presented some algorithms for real-world datasets, based on the work in \cite{userlevel} and \cite{tyagi}, which guarantee user-level $\epsilon$-DP, and also provided theoretical proofs of their performance trends. {However, the work \cite{dp_preprint} {focuses} exclusively on a single sample mean query to a dataset; in this work, our primary focus is on multiple (possible sequential) queries that involve the sample mean \emph{and variance} of disjoint subsets of a dataset.} 

{Our main contributions in this work are as follows:
\begin{enumerate}
	\item We initiate a study into the worst-case bias errors in the sample mean and sample variance that result from clipping the number of user contributions, possibly arbitrarily. In particular, we exactly characterize the worst-case errors (over all datasets) in the estimation of the sample mean and variance, under arbitrary clipping strategies.
	\item We exactly characterize the user-level sensitivities of estimators of the sample mean and variance, which are obtained by clipping the number of user contributions, possibly arbitrarily. The exact user-level sensitivity of the sample variance computed in this work is a non-trivial generalization of the exact  \emph{item-level} sensitivity of the sample variance (see \cite[Lemma A.2]{dorazio}), which is a strict improvement over the bound on \emph{item-level} sensitivity in \cite[p. 10]{dwork06} that is often taken as the standard in the DP literature. Our (fairly involved) mathematical analysis of the sensitivity of estimators of the variance is, interestingly, closely related to the analysis of the worst-case bias error.
	\item With the aid of our exact characterizations of the worst-case bias and sensitivities, we propose a novel, iterative algorithm, for improving the overall privacy loss under composition of several user-level DP mechanisms ({assuming an original, publicly known occupancy array}), each of which releases the sample mean and variance of disjoint subsets of a dataset, while maintaining the worst estimation error across the subsets. {We mention that the hyperparameters used by our clipping algorithm are the dataset-independent worst-case errors, which hence do not require additional privacy budget for private estimation, unlike the case in the works \cite{hyper1,hyper2,hyper3,hyper4}.}
	\item For reducing the worst-case estimation errors in any subset \emph{post the execution of our algorithm}, we propose a natural extension of the psuedo-user creation-based mean estimation algorithm in \cite{dp_preprint}.
	\item We evaluate the performance of our main algorithm and our pseudo-user creation-based procedure on real-world Intelligent Traffic Management System (ITMS) data from an Indian city and on synthetic datasets.
\end{enumerate}
}
 
 Our algorithm achieves the claimed improvement in privacy loss by \emph{suppressing}, or {completely removing}, the contributions of selected users in selected subsets, while not increasing the \emph{largest worst-case} error across all the subsets. 
 We emphasize that our algorithm can be applied more generally to the release of other statistics (potentially different from the sample mean and variance) of several disjoint subsets of the records in a dataset, {as long as exact, analytical characterizations of the sensitivities and worst-case errors can be derived. 
 }

\subsection{Organization of material}

The paper is organized as follows: Section \ref{sec:prelim} presents the problem formulation and recapitulates preliminaries on DP and user-level DP. Section \ref{sec:pseudo} contains a description of  the mechanisms of importance to this paper and presents an exact characterization of the (user-level) sensitivity of the sample variance function. Section \ref{sec:variance} exactly characterizes the worst-case errors in the estimation of sample mean and variance due to the suppression of selected records. Section \ref{sec:alg} then describes our main algorithm that suppresses user contributions in an effort to improve the privacy loss under composition. We then numerically evaluate the performance of our algorithm on {real-world and} synthetically generated datasets in terms of the privacy loss degradation, in Section \ref{sec:results}, and suggest a simple pseudo-user creation-based algorithm to improve the worst-case estimation error, over all grids. The paper is concluded in Section \ref{sec:conclusion} with some directions for future research.

\section{Preliminaries}
\input{prelim.tex}
\section{Mechanisms for Releasing DP Estimates}
\input{pseudo-user.tex}
\section{Worst-Case Errors in Estimation of Sample Mean and Variance}
\input{var.tex}
\section{An Error Metric and an Algorithm for Controlling Privacy Loss}
\input{alg-clip.tex}
\section{Numerical Results}
\input{numerics.tex}
\section{Conclusion}
\input{conclusion.tex}

\bibliographystyle{IEEEtran}
{\footnotesize
	\bibliography{references.bib}}

\appendices
\section{Proof of Lemma \ref{lem:app1}}
\input{app-lem-app1.tex}
\section{Proof of Lemma \ref{lem:app2}}
\input{app-lem-app2.tex}
\section{Proof of Proposition \ref{prop:varsens}}
\input{app-varsens.tex}
\section{On the Sensitivities Under a Special Clipping Strategy}
\input{app-array-av.tex}
\section{Proof of Theorem \ref{thm:varworst}}
\input{varworst.tex}
\end{document}

%% file: prelim.tex
\label{sec:prelim}
\subsection{Notation}
For a given $n\in \mathbb{N}$, the notation $[n]$ denotes the set $\{1,2,\ldots,n\}$ and the notation $[a:b]$ denotes the set $\{a,a+1,\ldots,b\}$, for $a,b\in \mathbb{N}$ and $a\leq b$. Given a length-$n$ vector ${u}\in \mathbb{R}^n$, we define $\lVert {u} \rVert_1:= \sum_{i=1}^n |u_i|$ to be the $\ell_1$-norm of the vector ${u}$.
We write $X\sim P$ to denote that the random variable $X$ is drawn from the distribution $P$. We use the notation $\text{Lap}(b)$ to refer to a random variable $X$ drawn from the zero-mean Laplace distribution with standard deviation $\sqrt{2}b$; its probability distribution function (p.d.f.) obeys
\[
f_X(x) = \frac{1}{2b}e^{-|x|/b}, \ x\in \mathbb{R}.
\]
\subsection{Motivation and Problem Setup}
{This work is motivated by the analysis of {spatio-temporal datasets}, such as datasets of traffic information, call records or credit card spends, which contain records of the data provided by sensors or cellular equipment in a city or a country. Each record catalogues, typically among other information, an identification of a user (say, via the licence plate or phone number), the location at which the data was recorded, a timestamp, and the actual data value itself, (say the speed, call duration, or AQI).} Most data analysis tasks on such datasets proceed as follows: first, in an attempt to obtain fine-grained information about  statistics in different areas of the city or country, the total area is divided into hexagon-shaped grids (see, e.g., Uber's Hexagonal Hierarchical Spatial Indexing System \cite{h3}, which provides an open-source library for such partitioning tasks). Next, the timestamps present in the data records are quantized (or binned) into timeslots of fixed duration (say, one hour). 

\subsection{Problem Formulation}
Let $\mathcal{L}$ denote the collection of all users in the dataset, and let $\mathcal{G}$ be the collection of grids or identifiers of subsets of the dataset (in this paper, we use ``grids'' and ``subsets'' interchangeably). We set $L := |\mathcal{L}|$ and $G := |\mathcal{G}|$. Furthermore, for each user $\ell\in \mathcal{L}$ and each grid $g\in \mathcal{G}$, we let $ m_{g,\ell}$ denote the (non-negative integer) number of samples contributed by user $\ell$ in the records corresponding to grid $g$. Now, for a given user $\ell\in \mathcal{L}$, let $m_\ell:= \sum_{g\in \mathcal{G}} m_{g,\ell}$ be the total number of samples contributed by user $\ell$ across all grids. Next, for every grid $g\in \mathcal{G}$, let 
\[
\mathcal{L}_g := \{\ell \in \mathcal{L}:\ m_{g,\ell} >0\}
\]
be the collection of users whose contributions constitute the data records corresponding to grid $g$. We let $m_g^\star$ denote the largest number of samples contributed by any user in grid $g\in \mathcal{G}$. Formally, $m_g^\star = \max_{\ell\in \mathcal{L}_g} m_{g,\ell}$. For every user $\ell\in \mathcal{L}$, let
\[
\mathcal{G}_\ell := \{g\in \mathcal{G}:\ m_{g,\ell}>0\} 
\]
be the collection of grids whose records user $\ell$ contributes to. In line with the previous notation, we set $L_g := |\mathcal{L}_g|$ and $G_\ell:= |\mathcal{G}_\ell|$. {Further, we assume that the sets $\mathcal{L}_g$, $g\in \mathcal{G}$, and $\mathcal{G}_\ell$, $\ell\in \mathcal{L}$, are publicly known. In other words, we assume that the occupancy array $\mathsf{A}\in \mathbb{N}^{L\times G}$ such that $\mathsf{A}_{\ell,g} = m_{g,\ell}$, for $\ell\in \mathcal{L}, g\in \mathcal{G}$, is publicly known.} Throughout this paper, we assume, without loss of generality, that $G_1\geq G_2\geq \ldots\geq G_L$.

Now, let $S_{g,\ell}$ denote the vector of samples contributed by user $\ell \in \mathcal{L}$ in grid $g\in \mathcal{G}$; more precisely, $S_{g,\ell}:= \left(S_{g,\ell}^{(j)}:\ j\in \left[m_{g,\ell}\right] \right)$. We assume that each $S_{g,\ell}^{(j)}$ is a non-negative real number that lies in the interval $[0,U]$, where $U$ is a fixed, {publicly known} {upper bound on the sample values (the results in this paper can also be extended to situations where the data samples can take negative, but bounded values)}. For most real-world datasets, the samples are drawn according to some unknown {joint distribution $P$ over all the samples contributed by the different users}, that is potentially non-i.i.d. (where i.i.d. stands for ``independent and identically distributed'') across samples and users. Our analysis is distribution-free in that we work with the worst-case errors in estimation over all datasets, in place of distribution-dependent error metrics such as the expected error (see, e.g., \cite[Sec. 1.1]{primer_stat} for a discussion).


We call the dataset consisting of the records contributed by users as $$\mathcal{D} = \left\{\left(\ell,\left\{S_{g,\ell}: g\in \mathcal{G}\right\}\right): \ell \in \mathcal{L}\right\}.$$ We let $\mathsf{D}$ denote the universe of all possible datasets with a given distribution of numbers of samples contributed by users across grids $\{m_{g,\ell}:\ell\in \mathcal{L}, g\in \mathcal{G}\}$.

The function that we are interested in is the length-$G$ vector $f: \mathsf{D}\to (\mathbb{R}^2)^G$, {each of the $G$ components of which} is a $2$-tuple of the sample average and the sample variance of samples in each grid. More precisely, we have
\begin{equation}
	\label{eq:f}
f(\mathcal{D}) = \left(f_g(\mathcal{D}): g\in \mathcal{G}\right),
\end{equation}
where $f_g: \mathsf{D}\to \mathbb{R}^2$ is such that
\begin{equation}
	\label{eq:fg}
f_g(\mathcal{D}) = 
\begin{bmatrix}
	\mu_g(\mathcal{D})\\
	\textsf{Var}_g(\mathcal{D})
\end{bmatrix}.
\end{equation}
Here,
\begin{equation}
	\label{eq:meang}
\mu_g(\mathcal{D}):= \frac{1}{\sum\limits_{\ell\in {^g\!\mathcal{L}}} m_{g,\ell}}\cdot \sum_{\ell\in { \mathcal{L}_g}} \sum_{j=1}^{m_{g,\ell}} S_{g,\ell}^{(j)}
\end{equation}
is the sample mean corresponding to grid $g$ and
\begin{equation}
	\label{eq:varg}
\textsf{Var}_g(\mathcal{D}) = \frac{1}{\sum_{\ell\in \mathcal{L}_g} m_{g,\ell}}\cdot \sum_{\ell\in \mathcal{L}_g} \sum_{j=1}^{m_{g,\ell}} \left(S_{g,\ell}^{(j)}-\mu_g(\mathcal{D})\right)^2
\end{equation}
is the sample variance corresponding to grid $g$. For the purposes of this work, one can equivalently think of $f(\mathcal{D})$ as a length-$2G$ vector, each of whose components is a scalar mean or variance. A central objective in user-level differential privacy is the private release of an estimate of $f$, without compromising too much on the accuracy in estimation. We next recapitulate the definition of user-level differential privacy \cite{continual}{, for the setting where the occupancy array $\mathsf{A}$ is known and fixed}.

\subsection{User-Level Differential Privacy}
\label{sec:dp}
Consider two datasets $\mathcal{D}_1 = \left\{\left({\ell},\left\{x_{g,\ell}: g\in \mathcal{G}\right\}\right):\ \ell\in \mathcal{L}\right\}$ and \  $\mathcal{D}_2 = $ $ \left\{\left(\ell,\left\{\tilde{x}_{g,\ell}: g\in \mathcal{G}\right\}\right):\ \ell\in \mathcal{L}\right\}$ consisting of the same users, {with  $\mathcal{D}_1$ and $\mathcal{D}_2$ having the same occupancy array $\mathsf{A}$}. {Note however, that for a fixed dataset (either $\mathcal{D}_1$ or $\mathcal{D}_2$), we allow different users to contribute different numbers of samples.} We let $\mathsf{D}$ be the universal set of such databases{, with a fixed occupany array $\mathsf{A}$}. We say that $\mathcal{D}_1$ and $\mathcal{D}_2$ are ``user-level neighbours'' {(with common occupancy array $\mathsf{A}$)} if there exists $\ell_0\in [L]$ such that $\left(x_{g,\ell_0}: g\in \mathcal{G}\right)\neq \left(\tilde{x}_{g,\ell_0}: g\in \mathcal{G}\right)$, with $\left(x_{g,\ell}: g\in \mathcal{G}\right)= \left(\tilde{x}_{g,\ell}: g\in \mathcal{G}\right)$, for all $\ell\neq \ell_0$. Clearly, datasets $\mathcal{D}_1$ and $\mathcal{D}_2$ differ in at most $m_{\ell_0}$ samples, where $m_{\ell_0}\leq m^\star$, with $m^\star:= \max_{\ell\in \mathcal{L}} m_\ell$.
\begin{definition}
	For a fixed $\epsilon>0$, a mechanism $M: \mathsf{D}\to \mathbb{R}^d$ is said to be user-level $\epsilon$-DP if for every pair of datasets $\mathcal{D}_1, \mathcal{D}_2$ that are user-level neighbours, and for every measurable subset $Y \subseteq \mathbb{R}^d$, we have that
	\[
	\Pr[M(\mathcal{D}_1) \in Y] \leq e^\epsilon \Pr[M(\mathcal{D}_2) \in Y].
	\]
\end{definition}
Next, we recall the definition of the user-level sensitivity of a function of interest.
\begin{definition}
	Given a function $\theta: \mathsf{D}\to \mathbb{R}^d$, we define its user-level sensitivity $\Delta_\theta$ as
	\[
	\Delta_\theta:= \max_{\mathcal{D}_1,\mathcal{D}_2\ \text{u-l nbrs.}} \left \lVert \theta(\mathcal{D}_1) - \theta(\mathcal{D}_2)\right \rVert_1,
	\]
	where the maximization is over datasets that are user-level neighbours.
\end{definition}


In this paper, we use the terms ``sensitivity'' and ``user-level sensitivity'' interchangeably. The next result is well-known and follows from standard DP results \cite[Prop. 1]{dwork06}\footnote{It is well-known that it is sufficient to focus on noise-adding DP mechanisms. The assumption that our mechanisms are \emph{additive-noise} or \emph{noise-adding} mechanisms is without loss of generality, since it is known that every privacy-preserving mechanism can be thought of as a noise-adding mechanism (see \cite[Footnote 1]{staircase} and \cite{opt}). Moreover, under some regularity conditions, for small $\epsilon$ (or
	equivalently, high privacy requirements), it is known that Laplace distributed noise is asymptotically optimal in terms of the magnitude of error in estimation \cite{staircase,opt}. {While it is possible to replace the Laplace mechanism in our paper with, say, the staircase mechanism \cite{staircase} in order to add less noise for $\epsilon$-DP, we persist with the Laplace mechanism for its simplicity in implementation.}}:

\begin{theorem}
	\label{thm:dp}
	For a function $\theta: \mathsf{D}\to \mathbb{R}^d$, the mechanism $M^{\text{Lap}}: \mathsf{D}\to \mathbb{R}^d$ defined by
	\[
	M^{\text{Lap}}(\mathcal{D}) = \theta(\mathcal{D})+Z,
	\]
	where $Z = (Z_1,\ldots,Z_d)$ is such that $Z_i\sim \text{Lap}(\Delta_\theta/\epsilon)$, is user-level $\epsilon$-DP.
\end{theorem}
For mechanisms as above, we also call $\epsilon$ as the ``privacy budget''. Furthermore, by standard results on the tail probabilities of Laplace random variables, we obtain the following bound on the estimation error due to the addition of noise for privacy:
\begin{proposition}
	\label{prop:esterror}
	For a given function $\theta:\mathsf{D}\to \mathbb{R}^d$ and for any dataset $\mathcal{D}_1$, we have that
	\[
	\Pr\left[\left \lVert M^{\text{Lap}}(\mathcal{D}_1) - \theta(\mathcal{D}_1)\right \rVert_1 \geq  \frac{\Delta_\theta \ln(1/\delta)}{\epsilon}\right]\leq \delta,
	\]
	for all $\delta\in (0,1]$.
\end{proposition}
In the following subsection, we shall discuss the overall privacy loss that results from the composition of several user-level $\epsilon$-DP mechanisms together.

\subsection{Composition of User-Level DP Mechanisms}
Recall that our chief objective in this work is the (potentially sequential, or adaptive) release of a fixed function (in particular, the sample mean and sample variance) of the records  in each grid, over all grids. The following fundamental theorem from the DP literature \cite[Cor. 3.15]{dworkroth} captures the worst-case privacy loss degradation upon composition of (user-level) DP mechanisms. For each $g\in \mathcal{G}$, let $M_g: \mathsf{D}\to \mathbb{R}^d$ be an $\epsilon_g$-DP algorithm that acts exclusively on those records from grid $g$. Further, let $M = \left(M_g:\ g\in \mathcal{G}\right)$ be the composition of the $G$ mechanisms above.
\begin{theorem}[Basic Composition Theorem]
	\label{thm:comp}
	We have that $M$ is user-level $\sum_{g\in \mathcal{G}} \epsilon_g$-DP.
\end{theorem}
It is well-known (see, e.g., \cite[Sec. 2.1]{steinkechapter}) that Theorem \ref{thm:comp} is tight, in that there exists a \emph{Laplace mechanism} (of the form in Theorem \ref{thm:dp}) that achieves a privacy loss of $\sum_{g\in \mathcal{G}} \epsilon_g$ upon composition. 
%

Observe from Theorem \ref{thm:comp} that in the case when $\epsilon_g = \epsilon$, for all $g\in \mathcal{G}$, we obtain an overall privacy loss of $G\epsilon$, upon composition. Clearly, when the number of grids $G$ is large, the overall privacy loss is  large, as well. 

We next present a simple improvement of the Basic Composition Theorem above that takes into account the fact that each mechanism $M_g$, $g\in \mathcal{G}$, acts only on the records in the grid $g$. {This improvement is a straightforward extension of the well-known ``parallel composition theorem" (see, e.g., \cite[Thm. 4]{pinq}) to the case when users contribute potentially more than one sample.} Let $\overline{\epsilon}:= \max_{\ell\in \mathcal{L}} \sum_{g\in \mathcal{G}_\ell} \epsilon_g$.

\begin{theorem}
	\label{thm:improvcomp}
	We have that $M$ is user-level $\overline{\epsilon}$-DP.
\end{theorem}
\begin{proof}
	Consider datasets $\mathcal{D}$ and $\mathcal{D}'$ that differ (exclusively) in the contributions of user $\ell\in \mathcal{L}$. Now, consider any measurable set $T = \left(T_1,\ldots,T_G\right)\subseteq \mathbb{R}^G$. For ease of reading, we let $M^{(g-1)}(\mathcal{D}):= (M_1(\mathcal{D}),\ldots,M_{g-1}(\mathcal{D})))$; likewise, we let $T^{(g-1)}:= (T_1,\ldots,T_{g-1})$.
	\begin{align*}
		&\frac{\Pr[M(\mathcal{D})\in T]}{\Pr[M(\mathcal{D}')\in T]}\\
		&= \frac{\prod_{g\in \mathcal{G}}\Pr[M_g(\mathcal{D})\in T_g\vert M^{(g-1)}(\mathcal{D})\in T^{(g-1)}]}{\prod_{g\in \mathcal{G}}\Pr[M_g(\mathcal{D}')\in T_g\vert M^{(g-1)}(\mathcal{D}')\in T^{(g-1)}]}\leq e^{\sum_{g\in \mathcal{G}_\ell} \epsilon_g},
	\end{align*}
	where the last inequality follows from the DP property of each mechanism $M_g$, $g\in \mathcal{G}$. The result then follows immediately.
\end{proof}
As a simple corollary, from our assumption that $G_1\geq G_2\geq \ldots \geq G_L$, we obtain the following result:
\begin{corollary}
	\label{cor:naive}
	When $\epsilon_g = \epsilon$, for all $g\in \mathcal{G}$, we have that $M$ is $G_1\epsilon$-DP.
\end{corollary}
In what follows, we shall focus on this simplified setting where the privacy loss $\epsilon_g$ for each grid $g$ is fixed to be $\epsilon>0$. Note that if $G_1$ is large, the privacy loss upon composing the mechanisms corresponding to the different grids is correspondingly large. 

A natural question that arises, hence is: can we improve the worst-case privacy loss (in the sense of Corollary \ref{cor:naive}) in such a manner as to preserve some natural notion of the worst-case error over all grids? In what follows, we shall show that for a specific class of (canonical) mechanisms, a notion of the worst-case error over all grids can be made precise and {exact, analytical expressions for this worst-case error} will then aid in the design of our algorithm that improves the privacy loss degradation by suppressing {or completely removing selected} user contributions.

We end this subsection with a remark. In the setting of \emph{item-level} DP, where each user contributes at most one sample, it follows from Theorem \ref{cor:naive} that the composition of mechanisms that act on \emph{disjoint} subsets of a dataset has the same privacy loss as that of any individual mechanism, i.e., $M$ is $\epsilon$-DP as well. In such a setting, it is not possible to improve on the privacy loss degradation by suppressing contributions of selected users.


The next section describes the mechanisms that will be of use in this paper; we refer the reader to \cite{userlevel, dp_preprint} for more user-level DP mechanisms for releasing sample means and their performance on real-world datasets.




%
%


%% file: pseudo-user.tex
\label{sec:pseudo}
In this section and the next, we focus our attention on a single grid $g\in \mathcal{G}$. For notational simplicity, we shall drop the explicit dependence of the notation (via superscripts) in Section \ref{sec:prelim} on $g$; alternatively, it is instructive to consider this setting as a special case of the setting in Section \ref{sec:prelim}, where $|\mathcal{G}| = 1$. In particular, we let $m_{g,\ell} =: m_\ell$, for all $\ell\in \mathcal{L}$, $\mathcal{L}_g =: \mathcal{L}$, $\mu_g =: \mu$, and $\textsf{Var}_g =: \textsf{Var}$. With some abuse of notation, we let $\mathcal{D}$ denote the dataset consisting of records in grid $g$ and let $\mathsf{D}$ denote the universal set of datasets with the distribution $\{m_\ell\}$ of user contributions.

We now describe two mechanisms for releasing user-level differentially private estimates of the sample mean and variance of a single grid. {We shall also explicitly identify analytical expressions for the sensitivities of the estimators used.}
\subsection{\textsc{Baseline}}
\label{sec:baseline}
Given the definitions $\mu$ and $\textsf{Var}$ as in \eqref{eq:meang} and \eqref{eq:varg}, the first mechanism, which we call \textsc{Baseline}, simply adds the right amount of Laplace noise to $\mu$ and $\textsf{Var}$ to ensure user-level $\epsilon$-DP. Formally, the \textsc{Baseline} mechanism $M_\text{b}:\mathsf{D}\to \mathbb{R}^2$ obeys
\[
M_\text{b}(\mathcal{D}) = \begin{bmatrix}
	M_{\mu, \text{b}}(\mathcal{D})\\
	M_{\textsf{Var}, \text{b}}(\mathcal{D})
\end{bmatrix},
\]
where
\[
M_{\mu, \text{b}}(\mathcal{D}) = \mu(\mathcal{D})+\text{Lap}(2\Delta_{\mu}/\epsilon),
\]
and
\[
M_{\textsf{Var},\text{b}}(\mathcal{D}) = \textsf{Var}(\mathcal{D})+\text{Lap}(2\Delta_{\textsf{Var}}/\epsilon).
\]
Note that the privacy budget for the release of each of the sample mean and variance is fixed to $\epsilon/2$, leading to $M_\text{b}$ being $\epsilon$-user-level DP, overall, by Theorem \ref{thm:comp}. {We mention that one can also consider mechanisms $M_{\mu, \text{b}}$ and $M_{\textsf{Var},\text{b}}$ with different privacy budgets $\epsilon_1>0$ and $\epsilon_2>0$ such that $\epsilon_1+\epsilon_2 = \epsilon$; one can then carry out an optimization over $\epsilon_1, \epsilon_2$ to obtain the best (or lowest) worst-case error (see Section \ref{sec:variance}).}
	
	\subsubsection{User-Level sensitivities of $\mu$ and $\textsf{Var}$}
	{Assuming that the privacy budgets of $M_{\mu, \text{b}}$ and $M_{\textsf{Var},\text{b}}$ are each $\epsilon/2$, we have from the definition of user-level sensitivity in Section \ref{sec:prelim} that}
\begin{equation}
\label{eq:deltamean}
\Delta_{\mu} = \frac{U\cdot \max_{\ell\in \mathcal{L}} m_\ell}{\sum\limits_{\ell\in {\mathcal{L}}}{m_\ell}}
= \frac{U\cdot m^\star}{\sum\limits_{\ell\in {\mathcal{L}}} m_\ell}.
\end{equation}
An explicit computation of the user-level sensitivity $\Delta_{\textsf{Var}}$ of \textsf{Var}, however, requires significantly more effort. The next proposition exactly identifies $\Delta_{\textsf{Var}}$.
\begin{proposition}
	\label{prop:varsens}
	We have that $$\Delta_{\textsf{Var}} = \begin{cases}
		\frac{U^2\ m^\star(\sum_\ell  m_\ell - m^\star)}{\left(\sum_\ell m_\ell\right)^2},\ \text{if } \sum_\ell m_\ell> 2m^\star,\\
		\frac{U^2}{4},\ \text{if $\sum_\ell m_\ell\leq 2m^\star$ and $\sum_\ell m_\ell$ is even},\\
		\frac{U^2}{4}\cdot\left(1-\frac{1}{(\sum_\ell m_\ell)^2}\right),\ \text{if $\sum_\ell m_\ell\leq 2m^\star$ and $\sum_\ell m_\ell$ is odd}.
	\end{cases}.$$
\end{proposition}
{The proof of Proposition \ref{prop:varsens} follows from a couple of helper lemmas. In what follows, we shall discuss these lemmas, whose proofs in turn are provided in Appendices \ref{app:lem-app-1} and \ref{app:lem-app-2}. The proof of Proposition \ref{prop:varsens} is then concluded in Appendix \ref{app:versens}. Before we do so, we examine some of its consequences. Importantly, we} obtain the following corollary on the sensitivity of the sample variance function in the item-level DP setting where each user contributes exactly one sample, i.e., when $m_\ell = 1$, for all $\ell\in \mathcal{L}$. We mention that this exact sensitivity expression for the item-level DP setting was derived in \cite[Lemma A.2]{dorazio}.
\begin{corollary}
	\label{cor:sensvaritem}
	In the setting of item-level DP, we have that for $L\geq 1$,
	$$\Delta_{\textsf{Var}} =
		\frac{U^2(L-1)}{L^2}.$$
\end{corollary}

On the other hand, the well-known upper bound on the sensitivity of the sample variance in \cite[p. 10]{dwork06} that is now standard for DP applications shows that in the item-level DP setting, $$\Delta_{\textsf{Var}} \leq \frac{8U^2}{L}.$$ Clearly, the exact sensitivity computed in \cite[Lemma A.2]{dorazio} and in Corollary \ref{cor:sensvaritem} is a \emph{strict} improvement over this bound, by a multiplicative factor of more than $8$, for all $L$.

Now, consider the expression in Proposition \ref{prop:varsens} above, for a \emph{fixed} $\sum_\ell m_\ell$. Suppose also that $\sum_\ell m_\ell>2m^\star$. 
Hence, for this range of $m^\star$ values, it is easy to argue that $h(m^\star):= m^\star(\sum_\ell  m_\ell - m^\star)$ is increasing in $m^\star$, implying that for a fixed value of $\sum_\ell m_\ell$, we have that $\Delta_{\textsf{Var}}$ is increasing in $m^\star$, in the regime where $\sum_\ell m_\ell>2m^\star$. {Furthermore, it can be argued that for a fixed positive real number $a$ we have that
\[
\overline{h}(a_1) = \frac{a_1(a-a_1)}{a^2}\leq \frac{1}{4},
\]
 for $0<a_1\leq a$. This then implies that} $\Delta_{\textsf{Var}}\leq U^2/4$, for all values of $\{m_\ell\}$, implying that $\Delta_{\textsf{Var}}$ is non-decreasing, overall, as $m^\star$ increases. In other words, a large value of $m^\star$ leads to a large sensitivity. In our next mechanism called \textsc{Clip}, which is the subject of Section \ref{sec:array-av}, we attempt to ameliorate this issue by clipping the number of contributions of each user in the grid, at the cost of some error in accuracy. 

	We shall now proceed to lay out the component lemmas that help prove Proposition \ref{prop:varsens}. Before we do so, we shall set up some notation. Recall from the definition of user-level sensitivity in Section \ref{sec:prelim} that 
	\[
	\Delta_{\textsf{Var}} = \max_{\mathcal{D} \sim \mathcal{D}'} \left \lvert \textsf{Var}(\mathcal{D}) - \textsf{Var}(\mathcal{D}')\right \rvert,
	\]
	where $\textsf{Var}$ is as in \eqref{eq:varg}, and the notation $\mathcal{D} \sim \mathcal{D}'$ refers to the fact that $\mathcal{D}$ and $\mathcal{D}'$ are user-level neighbours, for $\mathcal{D},\mathcal{D}'\in \mathsf{D}$. Moreover, without loss of generality, for the purpose of evaluating $\Delta_{\textsf{Var}}$, we can assume that $\textsf{Var}(\mathcal{D}')\leq \textsf{Var}(\mathcal{D})$ in the expression for $\Delta_{\textsf{Var}}$. Now, let $$\mathsf{D}_\text{max} = \left\{(\mathcal{D},\mathcal{D}'): (\mathcal{D},\mathcal{D}')\in \argmax_{\mathcal{D} \sim \mathcal{D}'} \left \lvert \textsf{Var}(\mathcal{D}) - \textsf{Var}(\mathcal{D}')\right \rvert\right\}$$ be the collection of pairs of neighbouring datasets that attain the maximum in the definition of $\Delta_{\textsf{Var}}$. In what follows, we shall exactly determine $\Delta_{\textsf{Var}}$ by identifying the structure of \emph{one} pair  $(\mathcal{D}_1, \mathcal{D}_2)\in 
	\mathsf{D}_\text{max}$ of neighbouring datasets.
	
	Suppose that $\mathcal{D}_1, \mathcal{D}_2$ as above differ (exclusively) in the sample values contributed by user $k\in [L]$. Let $\left\{S_\ell^{(j)}\right\}$ denote the samples in dataset $\mathcal{D}_1$ and $\left\{\tilde{S}_\ell^{(j)}\right\}$ denote the samples in dataset $\mathcal{D}_2$. Let $\nu$ and $\tilde{\nu}$ be respectively the sample means of $\left\{S_\ell^{(j)}\right\}$  and $\left\{\tilde{S}_\ell^{(j)}\right\}$. Let ${A}:= \left\{{S}_k^{(j)}:\ j\in [m_k]\right\}$ and $\tilde{A}:= \left\{\tilde{S}_k^{(j)}:\ j\in [m_k]\right\}$ be the samples contributed by user $k$ in $\mathcal{D}_1$ and $\mathcal{D}_2$, respectively. Further, let 
	\[
	{\nu}(A):= \frac{1}{m_k}\cdot \sum_{j=1}^{m_k} {S}_k^{(j)}
	\ \  \text{and}\ \ 
	{\nu}(\tilde{A}):= \frac{1}{m_k}\cdot \sum_{j=1}^{m_k} \tilde{S}_k^{(j)}
	\]
	be the means of the samples in $A$ and $\tilde{A}$, respectively. Similarly, let
	\[
	{\nu}(A^c):= \frac{1}{\sum_{\ell\neq k}m_\ell}\cdot \sum_{\ell\neq k} \sum_{j=1}^{m_\ell} {S}_\ell^{(j)}
	\ \  \text{and}\ \ 
	{\nu}(\tilde{A}^c):= \frac{1}{\sum_{\ell\neq k}m_\ell}\cdot \sum_{\ell\neq k} \sum_{j=1}^{m_\ell} \tilde{S}_\ell^{(j)},
	\]
	where we define $A^c$ to be those samples contributed by the users other than user $k$ in $\mathcal{D}_1$, and similarly, for $\tilde{A}^c$. By the definition of the datasets $\mathcal{D}_1$ and $\mathcal{D}_2$, we have that $A^c = \tilde{A}^c$ and hence $\tilde{\nu}(A^c) = \nu(A^c)$. We then have that the following lemma, whose proof is provided in Appendix \ref{app:lem-app-1}, holds.
	\begin{lemma}
		\label{lem:app1}
		There exists $(\mathcal{D}_1,\mathcal{D}_2)\in \mathsf{D}_\text{max}$ such that
		\[
		{\nu}(\tilde{A}) = {\nu}(A^c).
		\]
		Furthermore, we can choose $\tilde{S}_k^{(1)} = \ldots = \tilde{S}_k^{(m_k)} = {\nu}(\tilde{A})$, in $\mathcal{D}_2$.
	\end{lemma}
	From the proof of the lemma above, we obtain that there exist datasets $(\mathcal{D}_1$, $\mathcal{D}_2)\in \mathsf{D}_\text{max}$, such that
	\[
	\textsf{Var}({\mathcal{D}_2}) = \E\left[(\tilde{X}-{\nu}(A^c))^2\mid \tilde{X}\in {A}^c\right]\cdot \left(1- \frac{m_k}{\sum_\ell m_\ell}\right).
	\]
	Furthermore, for this choice of $\mathcal{D}_2$, we have $ \tilde{S}_k^{(j)} = {\nu}({A}^c)$, for all $j\in [m_k]$. The next lemma provides an alternative characterization of $\Delta_{\textsf{Var}}$, using our choice of datasets $\mathcal{D}_1$, $\mathcal{D}_2$.
	\begin{lemma}
		\label{lem:app2}
		We have that
		\begin{align*}
			\Delta_{\textsf{Var}} &= \max_{\mathcal{D}:\ {S}_\ell^{(j)} = \nu(A^c), \forall {S}_\ell^{(j)}\in A^c} \textsf{Var}(\mathcal{D})
		\end{align*}
	\end{lemma}
The proof of Lemma \ref{lem:app2} is provided in Appendix \ref{app:lem-app-2}. Note that the maximization in the expression in Lemma \ref{lem:app2} is essentially over $\nu(A^c)$ and the variables $\left\{S_j^{(\ell)}\right\}\in A$, with the constraint that ${S}_\ell^{(j)} = \nu(A^c)$, for all ${S}_\ell^{(j)}\in A^c$. It is easy to show that for a fixed choice of the variables $\left\{S_j^{(\ell)}\right\}\in A$, the expression in \eqref{eq:app3} is a quadratic function of $\nu(A^c)$, with a non-negative coefficient. Hence, the maximum over $\nu(A^c)$ of the expression in \eqref{eq:app3} is attained at a boundary point, i.e., at either $\nu(A^c) = 0$ or at $\nu(A^c) = U$. This observation then leads to a proof of Proposition \ref{prop:varsens} that is provided in Appendix \ref{app:versens}.

{In the next section, we shall describe another mechanism that constructs natural ``clipped'' estimators of the sample mean and variance, which we shall use in our algorithm that obtains gains in composition privacy loss, for fixed estimation error. We shall then explicitly identify the sensitivities of these estimators.}
\subsection{\textsc{Clip}}
\label{sec:array-av}
We proceed to describe a simple modification of the previous mechanism, which we call \textsc{Clip}, for releasing user-level differentially private estimates of $\mu$ and $\textsf{Var}$, by clipping (or suppressing) selected records. {We shall later use such mechanisms with special structure, where contributions of selected users are \emph{suppressed} or clipped entirely, to obtain improvements in the privacy loss under composition, for fixed estimation error.}
For $\ell\in \mathcal{L}$, we let $\Gamma_\ell\in [0:m_\ell]$ denote the number of contributions of user $\ell$ that have \emph{not} been clipped; without loss of generality, we assume that the set of indices of these samples is $[\Gamma_\ell]$. Further, we assume that $\sum_\ell \Gamma_{\ell}>0$. We use the notation $\Gamma^\star:= \max_{\ell\in \mathcal{L}} \Gamma_\ell$.

Given the dataset $\mathcal{D}$, we set
\begin{equation}
	\label{eq:meanarr}
	\mu_\text{clip}(\mathcal{D}) = \frac{1}{\sum_\ell \Gamma_\ell} \cdot \sum_{\ell=1}^L \sum_{j=1}^{\Gamma_\ell} S_\ell^{(j)}
\end{equation}
to be that estimator of the sample mean that is obtained by retaining only $\Gamma_\ell$ samples, for each user $\ell$.
Next, we set
\begin{equation}
	\label{eq:vararr}
		\textsf{Var}_{\text{clip}}(\mathcal{D}) = \frac{1}{\sum_\ell \Gamma_\ell}\cdot \sum_{\ell=1}^L \sum_{j=1}^{\Gamma_\ell} \left(S_\ell^{(j)}-\mu_\text{clip}(\mathcal{D})\right)^2
\end{equation}
to be an estimator of the sample variance that makes use of the previously computed estimator $\mu_\text{clip}(\mathcal{D})$ of the sample mean.


Our mechanism $M_\text{clip}: \mathsf{D}\to \mathbb{R}^2$ obeys
\begin{equation}
	\label{eq:Mclip}
M_\text{clip}(\mathcal{D}) = \begin{bmatrix}
	M_{\mu, \text{clip}}(\mathcal{D})\\
	M_{\textsf{Var}, \text{clip}}(\mathcal{D})
\end{bmatrix},
\end{equation}
where
\[
M_{\mu, \text{clip}}(\mathcal{D}) = \mu_\text{clip}(\mathcal{D})+\text{Lap}(2\Delta_{\mu_\text{clip}}/\epsilon),
\]
and
\[
M_{\textsf{Var}, \text{clip}}(\mathcal{D}) = \textsf{Var}_{\text{clip}}(\mathcal{D})+\text{Lap}(2\Delta_{\textsf{Var}_{\text{clip}}}/\epsilon).
\]
Here, $\Delta_{\mu_\text{clip}}$ and $\Delta_{\textsf{Var}_{\text{clip}}}$ are respectively the user-level sensitivities of the clipped mean estimator $\mu_\text{clip}$ and the clipped variance estimator $\textsf{Var}_\text{clip}$. As before, we assign a privacy budget of $\epsilon/2$ for each of the mechanisms $M_{\mu, \text{clip}}$ and $M_{\textsf{Var}, \text{clip}}$. Clearly, both these algorithms are $\epsilon/2$-user-level DP, from Theorem \ref{thm:dp}, resulting in the overall mechanism $M_\text{clip}$ being $\epsilon$-user-level DP, from Theorem \ref{thm:comp}.


By arguments similar to those in \cite[Sec. III.C]{dp_preprint}, we have that
\begin{equation}
	\label{eq:muarrsens}
\Delta_{\mu_\text{clip}}= \frac{U\ \Gamma^\star}{\sum_{\ell=1}^L \Gamma_\ell}.
\end{equation}
Furthermore, by analysis entirely analogous to the proof of Proposition \ref{prop:varsens}, we obtain the following lemma:
\begin{lemma}
	\label{lem:vararrsens}
We have that
\[
\Delta_{\textsf{Var}_\text{clip}} = 
\begin{cases}
	\frac{U^2\ \Gamma_\ell^\star(\sum_\ell \Gamma_\ell - \Gamma_\ell^\star)}{\left(\sum_\ell \Gamma_\ell\right)^2},\ \text{if } \sum_\ell \Gamma_{\ell}>2\Gamma^\star,\\
	\frac{U^2}{4},\ \text{if $\sum_\ell \Gamma_\ell\leq 2\Gamma^\star$ and $\sum_\ell \Gamma_\ell$ is even},\\
	\frac{U^2}{4}\cdot\left(1-\frac{1}{(\sum_\ell \Gamma_\ell)^2}\right),\ \text{if $\sum_\ell \Gamma_\ell\leq 2\Gamma^\star$ and $\sum_\ell \Gamma_\ell$ is odd}.
\end{cases}
\]
\end{lemma}

In Appendix \ref{app:array-av}, we show that for a special class of clipping strategies considered in \cite{dp_preprint}, the sensitivities $\Delta_{\mu_\text{clip}}$ and $\Delta_{\textsf{Var}_\text{clip}}$ are in fact at most the values of their \textsc{Baseline} counterparts $\Delta_\mu$ and $\Delta_\textsf{Var}$, respectively. {For such special clipping strategies, the mechanisms $M_{\mu, \text{clip}}$ and $M_{\textsf{Var}, \text{clip}}$ are also called as pseudo-user creation-based mechanisms \cite{dp_preprint}.}

In the next section, we focus more closely on the \textsc{Clip} mechanism and explicitly characterize the worst-case errors (over all datasets) due to clipping the contributions of users.

%% file: var.tex
\label{sec:variance}
In this section, we continue to focus on a single grid $g\in \mathcal{G}$. {We then formalize the notion of the \emph{worst-case} error (or worst-case bias) due to clipping incurred, over all datasets, by the \textsc{Clip} mechanism with an arbitrary choice $\Gamma_{\ell}\in [0:m_\ell]$, for $\ell \in \mathcal{L}$. With the aid of this definition, we shall explicitly derive analytical expressions for the worst-case clipping error for the sample mean and variance estimators in Section \ref{sec:array-av}.} The characterizations of worst-case errors will be of use in the design of our algorithm for improving the privacy loss degradation under composition, via the clipping (or suppression) of user contributions in selected grids. We now make the notion of the worst-case clipping error formal.

Consider the functions $\mu, \textsf{Var}$ that stand for the true sample mean and variance, and the functions $\mu_{\text{clip}}, \textsf{Var}_{\text{clip}}$ that stand for the sample mean and variance of the clipped samples, for some fixed values $\Gamma_{\ell}\in [0:m_\ell]$, where $\ell\in \mathcal{L}$. We now define
\begin{equation*}
	E_{\mu}(\mathcal{D}):= \lvert \mu(\mathcal{D}) - \mu_{\text{clip}}(\mathcal{D})\rvert
\end{equation*}
as the clipping error (or bias) for the mean on dataset $\mathcal{D}$, and
\begin{equation*}
	E_{\mu}:= \max_{\mathcal{D}\in \mathsf{D}} E_{\mu}(\mathcal{D})
\end{equation*}
as the worst-case clipping error for the mean. Likewise, we define
\begin{equation*}
	E_{\textsf{Var}}(\mathcal{D}):= \lvert \textsf{Var}(\mathcal{D}) - \textsf{Var}_{\text{clip}}(\mathcal{D})\rvert
\end{equation*}
as the clipping error for the variance on dataset $\mathcal{D}$, and
\begin{equation*}
	E_{\textsf{Var}}:= \max_{\mathcal{D}\in \mathsf{D}} E_{\textsf{Var}}(\mathcal{D})
\end{equation*}
as the worst-case clipping error for the variance.  The theorem below, which follows from \cite[Lemma V. 1]{dp_preprint} then holds.
\begin{theorem}
	\label{thm:meanworst}
	We have that 
	\[
	E_\mu = U\cdot \left(1 - \frac{\sum_\ell \Gamma_\ell}{\sum_\ell m_\ell}\right).
	\]
\end{theorem}
{While  \cite{dp_preprint} contained a proof of Theorem \ref{thm:meanworst} for the special case when $\Gamma_{\ell} = \min\{m,m_\ell\}$, for $\ell\in \mathcal{L}$ and for some fixed $m\in [m_\star,m^\star]$, we mention that such a statement holds for general values $\Gamma_{\ell}\in [0: m_\ell]$ as well -- the proof of Theorem \ref{thm:meanworst} hence follows directly from the proof of Lemma V. 1 in \cite{dp_preprint}.} Next, we characterize exactly the worst-case clipping error (or worst-case bias) for the variance. 
\begin{theorem}
	\label{thm:varworst}
	We have that $E_{\textsf{Var}} = 0$\ $\text{if\ \ $\Gamma_{\ell} = m_\ell$, for all $\ell\in \mathcal{L}$}$. Furthermore, if $\sum_\ell \Gamma_{\ell}< \sum_\ell m_\ell$, we have
	\[
	E_{\textsf{Var}} = 
	\begin{cases}
		\frac{U^2\cdot \sum_\ell \Gamma_\ell\cdot \sum_{\ell'} (m_{\ell'} - \Gamma_{\ell'})}{\left(\sum_\ell m_\ell\right)^2},\ \text{if } \sum_\ell m_\ell > 2\sum_\ell \Gamma_\ell,\\
		\frac{U^2}{4},\ \text{if $\sum_\ell m_\ell \leq 2\sum_\ell \Gamma_\ell$ and $\sum_\ell m_\ell$ is even},\\
		\frac{U^2}{4}\cdot\left(1-\frac{1}{(\sum_\ell m_\ell)^2}\right),\ \text{if\ $\sum_\ell m_\ell \leq 2\sum_\ell \Gamma_\ell$ and $\sum_\ell m_\ell$ is odd}.
	\end{cases}	
	\]
\end{theorem}
The proof of Theorem \ref{thm:varworst}, which, interestingly, relies on arguments made in the proof of Proposition \ref{prop:varsens}, is provided in Appendix \ref{app:varworst}. {An exploration of a unified derivation of sensitivities and worst-case errors for other statistics is of broad interest and can be explored in future work.}

%% file: alg-clip.tex
\label{sec:alg}
In this section, we return to our original problem of releasing the sample means and variances of different grids, possibly sequentially. We present our algorithm that seeks to control the privacy loss of a certain user-level DP mechanism for jointly releasing the sample mean and variance of all grids in the city, by suppressing user contributions. As we shall see, the individual mechanisms for each grid simply add a suitable amount of Laplace noise that is tailored to the sensitivity of the functions in the grid \emph{post} clipping. Our algorithm hence crucially relies on the analyses of the sensitivity and the worst-case clipping error of the \textsc{Clip} mechanism in Sections \ref{sec:array-av} and \ref{sec:variance}.

\subsection{An Error Metric for Worst-Case Performance}

We shall first formally define a notion of the worst-case error of any mechanism $M = (M_g:\ g\in \mathcal{G})$, over all datasets, \emph{and over all grids}. Our algorithm will then follow naturally from these definitions.

Formally, consider a mechanism $M_g^\theta: \mathsf{D}\to \mathbb{R}^d$, for $g\in \mathcal{G}$, for the user-level differentially private release of a statistic $\theta_g: \mathsf{D}\to \mathbb{R}^d$ of the records in grid $g$. Suppose that $M_g^\theta$ obeys 
\begin{equation}
	\label{eq:Mbar}
	M_g^\theta(\mathcal{D}) = \overline{\theta}_g(\mathcal{D})+\overline{Z},
\end{equation}
for some estimate $\overline{\theta}_g$ of $\theta_g$, such that the user-level sensitivity of $\overline{\theta}_g$ is $\Delta_{\overline{\theta}_g}$. Recall that the assumption that $M_g^\theta$ is a noise-adding mechanism is without loss of generality. Also, in \eqref{eq:Mbar}, we have that $\overline{Z}$ is a length-$d$ vector with $\overline{Z}_i\sim \text{Lap}\left(\Delta_{\left(\overline{\theta}_g\right)_i}/\epsilon\right)$, for each coordinate $i\in [d]$. Note that we work with the class of mechanisms that add Laplace noise tailored to the sensitivities of each grid, individually, since an explicit computation of the user-level sensitivity of the vector $f$ in \eqref{eq:f} (across all grids) is quite hard, thereby implying the necessity of loose bounds on the amount of noise added, when this notion of user-level sensitivity is used.

Now, consider the mechanism $M_\theta$ that consists of the composition of the mechanisms $M_g^\theta$, over $g\in \mathcal{G}$, i.e., $M^\theta = \left(M_g^\theta:\ g\in \mathcal{G}\right)$. In many settings of interest, a natural error metric for such a composition of mechanisms acting on different grids is the \emph{largest}  \emph{worst-case} estimation error among all the grids. 

Now, given a mechanism $M_g^\theta$ as in \eqref{eq:Mbar}, we define its \emph{worst-case} estimation error as
\begin{equation}
	\label{eq:eg}
E_g:= \sum_{i\in [d]} \max_{\mathcal{D}\in \mathsf{D}} \Big \lvert \left(\theta_g\right)_i(\mathcal{D}) - \left(\overline{\theta}_g\right)_i(\mathcal{D}) \Big\rvert + \E[\lVert\overline{Z}\rVert].
\end{equation}
Finally, we define the error metric $E$ of the mechanism $M_\theta$ to be the \emph{largest} worst-case estimation error among all the grids, i.e., $$E := \max_{g\in \mathcal{G}} E_g.$$

We now describe our algorithm for reducing the privacy loss under composition, which makes use of a specialization of the definitions in this section to the case when the mechanisms $M_g^\theta$ are one of $M_\text{b} = M_g^\text{b}$ (corresponding to \textsc{Baseline}) or $M_\text{clip} = M_g^\text{clip}$ (corresponding to \textsc{Clip}).

\subsection{An Algorithm for Suppressing User Contributions}

The algorithm discussed in this section results in a simple improvement of Theorem \ref{thm:comp} that takes into account the structure of the queries. We mention that query-dependent composition results are also known for, say, histogram queries (see \cite[Prop. 2.8]{vadhan}). Consider the \textsc{Baseline} mechanisms $M_g^\text{b} = \left[M_{g,\mu}^\text{b},  M_{g,\textsf{Var}}^\text{b}\right]^{\top}$, as defined in Section \ref{sec:baseline}, for estimating the statistics $\mu_g$ and $\textsf{Var}_g$, for each grid $g$ of a given dataset, with $\mathsf{M}_\text{b} = \left(M_g^\text{b}:\ g\in \mathcal{G}\right)$. Observe that initially, for any grid $g$, we have
\begin{align}
	E_g &= \E\left[\lvert\text{Lap}(2\Delta_{\mu_g}/\epsilon)\rvert\right] +  \E\left[\lvert\text{Lap}(2\Delta_{\textsf{Var}_g}/\epsilon)\rvert\right] \notag\\
	&= \frac{2U\ m_g^\star}{\epsilon\cdot\sum_{\ell\in\, \mathcal{L}_g} m_{g,\ell}}+\frac{2U^2\ m_g^\star\left(\sum_{\ell\in\, \mathcal{L}_g}m_{g,\ell} - m_g^\star\right)}{\epsilon\cdot \left(\sum_{\ell\in\, \mathcal{L}_g} m_{g,\ell}\right)^2}, \label{eq:alg1}
\end{align}
where the last equality follows from \eqref{eq:deltamean} and Proposition \ref{prop:varsens}. As defined earlier, we have $E := \max_{g\in \mathcal{G}} E_g.$ From Corollary \ref{cor:naive}, we notice that in order to improve the privacy loss upon composition, we must seek to reduce $G_1$, or the {largest} number of grids that any user ``occupies''. Our aim is to accomplish this reduction in such a manner as to not increase the worst-case error $E$. \footnote{We mention that our algorithm can be executed with \emph{any} bound $E$ on the worst-case error of each grid and not just $E = \max_{g\in \mathcal{G}} E_g.$} 
\subsubsection{Description of the iterative procedure}
Our algorithm proceeds in stages, at each stage suppressing \emph{all} the contributions of those users that occupy the largest number of grids, in selected grids that these users occupy. Clearly, since the objective is to not increase $E$, for each such user, we suppress his/her contributions in that grid which has the smallest overall (that is the sum of errors due to bias and due to the noise added for privacy; see \eqref{eq:eg}) error \emph{post suppression}. We emphasize that our algorithm, being iterative in nature, is not necessarily optimal in that it does not necessarily return the lowest possible privacy loss degradation factor for a fixed worst-case error $E$. { We reiterate that the hyperparameters that are the errors computed at each stage of the algorithm are independent of the samples present (indeed, they depend only on the occupancy array $\mathsf{A}$ and the publicly known parameter $U$), and hence their calculation does not need additional privacy budget, unlike those considered in \cite{hyper1,hyper2,hyper3,hyper4}.} Note also that while the worst-case error (over all grids) $E$ is fixed at the start of the algorithm and is maintained as an invariant throughout its execution, the individual errors corresponding to each grid could potentially increase due to the suppression of user contributions. We let $E_g^{(0)} := E_g$, for each grid $g\in \mathcal{G}$. 

For each step $t\geq 1$ in our algorithm, we pick the user(s) that occupy the largest number of grids. Define 
\[
\mathsf{L}^{(t)}:= \{\ell\in \mathcal{L}:\ G_\ell \geq G_j,\ \forall j\in \mathcal{L}\}
\]
as the set of users in the first step of our algorithm that occupy the largest number of grids. The superscript `$(t)$' denotes the fact that the algorithm is in stage $t$ of its execution. Recall from our assumption that $G_1\geq G_j$, for any user $j\in \mathcal{L}$, and hence, in stage $1$, we have user $1\in \mathsf{L}^{(1)}$. We sort the users in $\mathsf{L}^{(t)}$
in increasing order of their indices, as $ \ell_1<\ell_2<\ldots<\ell_{|\mathsf{L}^{(t)}|}$.

Now, for each user $\ell \in \mathsf{L}^{(t)}$, starting from user $\ell_1$, we calculate the worst-case error that could result in each grid he/she occupies by potentially suppressing his/her contributions entirely. More precisely, for each grid $g\in \mathcal{G}_\ell$, we set $m_{g,\ell} = 0$, and recompute the values of $\Delta_{\mu_g}$ and $\Delta_{\textsf{Var}_g}$. In particular, following the definitions in Section \ref{sec:array-av}, we note that after suppression in grid $g$, we have $\Gamma_{g,\ell} = 0$ and $\Gamma_{g,\ell'} = m_\ell$, for $\ell'\neq \ell$, with $\Gamma^\star_g =  \max\limits_{\ell'\in\, \mathcal{L}_g:\ \ell'\neq\ell} m_{g,\ell'}$. Thus,  \eqref{eq:deltamean} and Proposition \ref{prop:varsens}, can be used to compute the sensitivities of the new sample mean and sample variance in grid $g$, which we denote as $\Delta_{g,\mu_\text{clip}}(\ell)$ and $\Delta_{g,\textsf{Var}_\text{clip}}(\ell)$, respectively.

Moreover, such a suppression of the contributions of user $\ell\in \mathsf{L}^{(1)}$ in grid $g$ introduces some \emph{worst-case} clipping errors in the computation of $\mu$ and $\textsf{Var}$, which we call $E_{g,\mu}(\ell)$ and $E_{g,\textsf{Var}}(\ell)$, respectively. The exact magnitude of these clipping errors incurred can be computed using Theorems \ref{thm:meanworst} and \ref{thm:varworst}, using the same values of $\Gamma_{g,\ell'}$ and $\Gamma^\star_g$ as described above, for $\ell'\in \mathcal{L}_g$. Finally, following \eqref{eq:eg}, we compute the overall worst-case error in grid $g$, post the suppression of the contributions of user $\ell$ as
\begin{equation}
	\label{eq:egl}
E_g(\ell) = E_{g,\mu}(\ell)+E_{g,\textsf{Var}}(\ell)+ \frac{2\ \Delta_{g,\mu_\text{clip}}(\ell)}{\epsilon}+\frac{2\ \Delta_{g,\textsf{Var}_\text{clip}}(\ell)}{\epsilon}.
\end{equation}

After computing the worst-case errors $E_g(\ell)$ that could result in each grid $g\in \mathcal{G}_\ell$ due to the potential suppression of the contributions of user $\ell$ in grid $g$, we identify one grid 
\begin{equation}
	\label{eq:algl}
\mathsf{g}(\ell)\in { \argmax}_{g\in \mathcal{G}_\ell} E_g(\ell)
\end{equation}
and the corresponding error value $E_{\mathsf{g}(\ell)}(\ell)$. In the event that $E_{\mathsf{g}(\ell)}(\ell)\leq E$, where $E$ is the original worst-case error, we proceed with clipping (or suppressing) {all} the contributions of user $\ell$ in grid $\mathsf{g}(\ell)$. In particular, we update $\mathcal{L}_{\mathsf{g}(\ell)} \gets \mathcal{L}_{\mathsf{g}(\ell)}\setminus \{\ell\}$ and $\mathcal{G}_\ell \gets \mathcal{G}_\ell\setminus \{{\mathsf{g}(\ell)}\}$. We recompute $G_\ell:= \lvert \mathcal{G}_\ell\rvert$ and
the above procedure, starting from \eqref{eq:algl}, is then repeated for all users $\ell\in \mathsf{L}^{(t)}$.

Else, if $E_{\mathsf{g}(\ell)}(\ell)> E$, we reset $ \Gamma_{{\mathsf{g}(\ell)},\ell}$ to its original value at the start of the iteration and we halt the execution of the algorithm. We then return the value $K:=\max_{\ell\in \mathcal{L}} G_\ell$ as the final privacy loss degradation factor. Pseudo-code for the \textsc{Clip-User} procedure is shown as Algorithm \ref{alg:multi-hat}. Note that, by design, the algorithm \textsc{Clip-User} maintains the worst-case error across grids as $E$, at every stage of its execution.

   \begin{algorithm}
   	\caption{Suppressing user contributions}
   	\begin{algorithmic}[1]
   		\Procedure {\textsc{Clip-User}}{$\mathcal{D}$}
   		\State For each $g\in \mathcal{G}$, compute $E_g$ as in \eqref{eq:alg1}.
   		\State Compute $E = \max_{g\in \mathcal{G}} E_g$.
   		\State Set \texttt{Halt} $\gets$ \texttt{No} and $t\gets 1$.
   		\While {\texttt{Halt} $=$ \texttt{No}}
   		\State Compute $\mathsf{L}^{(t)} =  \{\ell\in \mathcal{L}:\ G_\ell \geq G_j,\ \forall j\in \mathcal{L}\}$.
   		\For{$\ell\in \mathsf{L}^{(t)}$}
   		\For{$g\in \mathcal{G}_\ell$}
   			\State Set $\Gamma_{g,\ell} = 0$.
   			\State Compute error $E_g(\ell)$ as in \eqref{eq:egl}.
   		\EndFor
   			\State Pick $\mathsf{g}(\ell)\in \argmin_{g\in \mathcal{G}_\ell} E_g(\ell)$.
   		\If{$E_{{\mathsf{g}(\ell)}}(\ell)> E$}
   		\State Set \texttt{Halt} $=$ \texttt{Yes}
   		\State Reset $\Gamma_{g,\ell}$ to $m_{g,\ell}$, for all $g\in \mathcal{G}_\ell$.
   		\State \textbf{break}
   		\Else
   		\State Restore $\Gamma_{g,\ell}$ to ${m_{g,\ell}}$, for all $g\in \mathcal{G}_\ell\setminus \{\mathsf{g}(\ell)\}$.
   		\State Update $\mathcal{G}_\ell\gets \mathcal{G}_\ell\setminus \{\mathsf{g}(\ell)\}$ and $\mathcal{L}_g\gets \mathcal{L}_g\setminus \{\ell\}$.
   		\State Compute $G_\ell$, for all $\ell \in \mathcal{L}$.
   		\EndIf
   		\EndFor
   		\If{\texttt{Halt $=$ \texttt{Yes}}}
   		\textbf{break}
   		\Else
   		\State Set $t\gets t+1$.
   		\EndIf
   		\EndWhile
   		\State Return $K \gets \max_{\ell\in \mathcal{L}} G_\ell$.
   		\EndProcedure
   	\end{algorithmic}
   	\label{alg:multi-hat}
   \end{algorithm}
   
%

\subsubsection{DP release of statistics post execution of Algorithm \ref{alg:multi-hat}}
   Given the distribution $\{{m_{g,\ell}}\}$ of user contributions post the execution of \textsc{Clip-User}, we release user-level differentially private estimates of the sample means $\mu_g(\mathcal{D})$ and sample variances $\textsf{Var}_g(\mathcal{D})$, for $g\in \mathcal{G}$, by using a version of the \textsc{Clip} mechanism for each grid, as discussed in Section \ref{sec:array-av}. More precisely, for each grid $g$, we compute the values $\{\Gamma_{g,\ell}\}$ of user contributions post the suppression of user contributions in Algorithm \ref{alg:multi-hat}, and release $M_\text{clip, post}(\mathcal{D}) = M_\text{clip}(\mathcal{D})$ as in \eqref{eq:Mclip}. The following proposition then holds, similar to Corollary \ref{cor:naive}.
   
   \begin{proposition}
   		\label{prop:alg}
   	When $\epsilon_g = \epsilon$, for all $g\in \mathcal{G}$, we have that $M_\text{clip, post}$ is $K\epsilon$-DP, with a maximum worst-case error $E$ over all grids.
   \end{proposition}

{Furthermore, post the execution of the \textsc{Clip-User} procedure, it may be desirable to reduce the worst-case estimation error across grids further, by choosing a certain strategy for clipping user contributions. We outline such a technique in the following subsection.}
   
   \subsubsection{{Improving worst-case error post  execution of Algorithm \ref{alg:multi-hat}}}
   \label{sec:tkde}
   Now that we have (potentially) reduced the expected privacy loss degradation via the execution of \textsc{Clip-User}, while maintaining the worst-case error across grids as $E$, we discuss a simple {pseudo-user creation-based} strategy, drawing on \cite{dp_preprint}, which seeks to reduce this worst-case error across grids. Let $\{\Gamma_{g,\ell}: g\in \mathcal{G}, \ell \in \mathcal{L}_g\}$ denote the distribution of user contributions across grids, for a fixed instantiation of user contributions \emph{post suppression} via \textsc{Clip-User}. Here, $\mathcal{L}_g$ denotes the set of users with non-zero contributions in grid $g$, \emph{post} the execution of \textsc{Clip-User}. 
   
   In an attempt to reduce the worst-case error across grids further, we clip the contributions of \emph{all} users in a grid $g$ to some value $m\in [\Gamma_{g,\star}:\Gamma_g^\star]$, where $\Gamma_{g,\star}:= \min_{\ell \in \mathcal{L}_g} \Gamma_{g,\ell} $ and $\Gamma^\star_g:= \max_{\ell \in \mathcal{L}_g} \Gamma_{g,\ell} $. More precisely, for any fixed grid $g$, we pick the first $\overline{\Gamma}_{g,\ell}$ contributions of each user $\ell\in \mathcal{L}_g$, where $ \overline{\Gamma}_{g,\ell}:= \min\{\Gamma_{g,\ell}, m\}$, for some $m\in [\Gamma_{g,\star}:\Gamma_g^\star]$. This corresponds to using a \emph{pseudo-user} creation-based clipping strategy \cite{dp_preprint}, as mentioned in Section \ref{sec:array-av}.
   
   We then compute the sensitivities $\overline{\Delta}_{g,\mu_\text{clip}}$ and $\overline{\Delta}_{g,\textsf{Var}_\text{clip}}$ of the resultant clipped estimators of the sample mean and variance, respectively, using \eqref{eq:muarrsens} and Lemma \ref{lem:vararrsens} and the above values of $\{\overline{\Gamma}_{g,\ell}: \ell \in \mathcal{L}_g\}$. We also compute the clipping errors (or bias) introduced, which we call $\overline{E}_{g, {\mu}}$ and $\overline{E}_{g, \textsf{Var}}$, using Theorems \ref{thm:meanworst} and \ref{thm:varworst}, with $\{ \overline{\Gamma}_{g,\ell}: \ell \in \mathcal{L}_g\}$ corresponding to the clipped user contributions and $\{{m_{g,\ell}}\}$ corresponding to the original user contributions. Here, note that we use $ \overline{\Gamma}_{g,\ell} = 0$ for those users $\ell \in \mathcal{L}$ with ${m_{g,\ell}} >0$ and $\Gamma_{g,\ell} = 0$. We then set
   \[
   \overline{E}_g(m):= \overline{E}_{g, {\mu}}+\overline{E}_{g, \textsf{Var}}+\frac{2\, \overline{\Delta}_{g,\mu_\text{clip}}}{\epsilon}+\frac{2\, \overline{\Delta}_{g,\textsf{Var}_\text{clip}}}{\epsilon}
   \]
   as the overall error post pseudo-user creation-based clipping in grid $g$, corresponding to a fixed value of $m$. Note that the errors involving the sensitivity terms correspond to a mechanism that adds Laplace noise to each of the clipped mean and variance functions, tuned to the sensitivities $\overline{\Delta}_{g,\mu_\text{clip}}$ and $\overline{\Delta}_{g,\textsf{Var}_\text{clip}}$, respectively, with privacy loss parameter set to be $\epsilon/2$. {As mentioned in Section \ref{sec:baseline}, the choice of each privacy loss parameter being $\epsilon/2$ can potentially be further optimized by practitioners, taking into account the values of $\overline{\Delta}_{g,\mu_\text{clip}}$ and $\overline{\Delta}_{g,\textsf{Var}_\text{clip}}$.} We then compute
   \begin{equation}
   	\label{eq:tkde1}
   \overline{E}_g:= \min_{m\in [\Gamma_{g,\star}\,:\,\Gamma_g^\star]} \overline{E}_g(m),
   \end{equation}
   and repeat these computations for each grid $g\in \mathcal{G}$. Finally, we set 
   \[
   \overline{E} = \overline{E}_\epsilon:= \max_{g\in \mathcal{G}} \overline{E}_g
   \]
   to be the new worst-case error across all grids.
   
  {In the next section, we empirically evaluate the performance of Algorithm \ref{alg:multi-hat}, via the gains in privacy loss under composition, and the pseudo-user creation-based procedure described in this subsection, on real-world and synthetic datasets.}

%% file: numerics.tex
\label{sec:results}
In this section, we test the performance of \textsc{Clip-User} on {real-world and} synthetically generated datasets, via the privacy loss degradation $K = K_\epsilon$ obtained at the end of its execution. {We also provide experimental results on potential improvements in the worst-case error obtained at the end of the execution of \textsc{Clip-User}, using the pseudo-user creation-based procedure described in Section \ref{sec:tkde}.} We first describe our experimental setup and then numerically demonstrate the improvements obtained in the privacy loss degradation factor by running \textsc{Clip-User} on these synthetic datasets.
\subsection{Experimental Setup}
\label{sec:experiment}
Since this work concentrates on \emph{worst-case} errors in estimation, it suffices to specify a dataset $\mathcal{D}$ by simply the collection $\{m_{g,\ell}:\ell \in \mathcal{L},\ g\in \mathcal{G}\}$ of user contributions across grids. {We shall first present results obtained from the execution of Algorithm \ref{alg:multi-hat} and the procedure in Section \ref{sec:tkde} on real-world Intelligent Traffic Management System (ITMS) data from an Indian city. Next, we shall generate a synthetic dataset, the distributions of the user contributions in which are picked to model datasets with a single ``heavy-hitter'' user, who contributes the largest number of samples in each grid.}

\subsubsection{Real-world ITMS dataset}
{The ITMS dataset that we consider contains records of the data provided by IoT sensors deployed in an Indian city containing, among other information, the license plates of buses, the location at which the data was recorded, a timestamp, and the instantaneous speed of the bus. {The upper bound $U$ on the speeds of the buses is taken as $65$ (in km/hr), based on the maximum possible reading of the sensors installed on public buses in certain Indian cities.} For the purpose of analysis, we divide the total area in the city of interest into hexagon-shaped grids, using Uber's Hexagonal Hierarchical Spatial Indexing System (H3) \cite{h3}. Furthermore, we quantize the timestamps present in the data records into 1 hour timeslots. We focus on the 9 AM--10 AM timeslot on a fixed day and pick those data records that pertain to the $50$ hexagons (or grids) that contain the largest total number of user contributions. We hence seek to privately release the sample means and variances of speeds of the buses in the chosen grids. Here, we have that $G=50$ and the number of users who contribute at least one sample to some grid among those chosen is $L = 223$. Furthermore, the largest number of grids that any user ``occupies'' is $G_1 = 11$.}

\subsubsection{Synthetic dataset generation}

To this end, we work with the following distribution on the values $\{m_{g,\ell}\}$, which we believe is a reasonable, although much-simplified, model of real-world traffic datasets. We fix a number of grids $G=12$ and a number of users $L = 2^{12}-1 = 4095$. {The upper bound $U$ on the speed values is taken to be $65$, as before.}
\begin{enumerate}[label = \roman*.]
	\item \textbf{User Occupancies}: We index the users $\ell \in \mathcal{L}$ from $1$ to $L$. Any user $\ell\in [2^j:2^{j+1}-1]$ occupies (or, has non-zero contributions in) exactly $G-j = 12-j$ grids, where $j\in [0:G-1]$. It is clear that in this setting, we have $G_1\geq G_2\geq \ldots \geq G_L$.
	
	Now, consider any user $\ell$ that occupies $k$ grids. We identify these $k$ grids among the $G$ overall grids by sampling a subset of $\mathcal{G}$ of cardinality $k$, uniformly at random.
	
	\item \textbf{Number of contributions}: For a user $\ell$ that occupies grids $g_1,\ldots, g_k$, for $k$ fixed as above, we sample the number of his/her contributions in grid $g_i$, $i\in [k]$ as $m_{g_i,\ell} \sim \text{Geo}(q)$, where $\text{Geo}(q)$ denotes the geometric distribution with parameter $q\in [0,1]$. In particular,
	\[
	\Pr[m_{g_i,\ell} = m] = q\cdot (1-q)^{m-1},\ m\in \{1,2,\ldots\}.
	\]
	\item \textbf{Scaling the maximum contributions}: For each grid $g\in \mathcal{G}$, we identify a single user $\ell \in \argmax_{\ell'\in \mathcal{L}_g} m_{g,\ell'}$ and scale his/her number of contributions as ${m_{g,\ell}} \gets (1+\gamma){m_{g,\ell}}$, for a fixed $\gamma>0$.
\end{enumerate}
We mention that Step 3 above is carried out to model most real-world datasets where there exists one {``heavy-hitter''} user who contributes more samples than any other user, in each grid. Furthermore, note that the actual speed samples $\{S_{g,\ell}\}$ contributed by users across grids could be arbitrary, but these values do not matter in our analysis, since we work with the worst-case estimation errors.

\subsubsection{Estimating Expected Privacy Loss Degradation}

For the real-world ITMS dataset, we simply execute the \textsc{Clip-User} algorithm and numerically compute the privacy loss degradation factor $K_\epsilon$, for each $\epsilon\in [0.1,2]$. We then set $P_\epsilon = K_\epsilon\epsilon$ and run the pseudo-user creation-based algorithm described in Section \ref{sec:tkde}, post the execution of \textsc{Clip-User}. We let $\overline{E}_\epsilon$ to be the worst-case error across grids after the execution of the pseudo-user creation-based algorithm.

{Since the $\{m_{g,\ell}\}$ values in the synthetic datasets are chosen randomly, we compute Monte Carlo estimates of the privacy loss under composition after the execution of \textsc{Clip-User} and the worst-case error post the execution of the pseudo-user creation-based algorithm.} More precisely, for a fixed $\gamma, q$, we draw $10$ collections of (random) $\{m_{g,\ell}\}$ values. On each such collection of values, representing a dataset $\mathcal{D}$, we execute \textsc{Clip-User} and compute the privacy loss degradation factor $K_\epsilon$ for $\epsilon \in [0.1,1]$. We mention that in our implementation of \textsc{Clip-User}, we refrain from clipping user contributions in that grid $\mathsf{g} = \argmin_{g\in \mathcal{G}} E_g$, for $E_g$ as in \eqref{eq:alg1}. As an estimate of the expected privacy loss degradation for the given $\gamma, q$ parameters, we compute the Monte-Carlo average
\[
\widehat{P}_\epsilon:= \frac{1}{10}\sum_{i=1}^{10}K^{(i)}_\epsilon\epsilon,
\]
where the index $i\in [10]$ denotes a sample collection of $\{m_{g,\ell}\}$ values as above, with $K^{(i)}_\epsilon$ denoting the privacy loss degradation returned by \textsc{Clip-User} for these values.

As before, for our simulations, for a fixed $\gamma, q$, we draw $10$ collections of (random) $\{m_{g,\ell}\}$ values. On each such collection of values, we execute \textsc{Clip-User} and the pseudo-user creation-based clipping strategy above for $\epsilon \in [0.1,1]$. As an estimate of the expected worst-case error across grids \emph{post} the execution of \textsc{Clip-User}, for the given $\gamma, q$ parameters, we compute the Monte-Carlo average
\[
\widehat{\overline{E}}_\epsilon:= \frac{1}{10}\sum_{i=1}^{10}\overline{E}^{(i)}_\epsilon,
\]
where the index $i\in [10]$ denotes a sample collection of $\{m_{g,\ell}\}$ values as above, with $\overline{E}^{(i)}_\epsilon$ denoting the worst-case error across grids for these values.

\subsection{Simulations}
Given the experimental setup described in the previous section, we now provide simulations that demonstrate the performance of \textsc{Clip-User} and the pseudo-user creation-based clipping strategy with regard to estimates of the expected privacy loss degradation and the expected worst-case error across grids. 

{For the real-world ITMS dataset, Figure \ref{fig:itms1} shows a plot of $P_\epsilon$ against $\epsilon$, as $\epsilon$ ranges from $0.1$ to $2$. We compare this plot with a plot of the original cumulative privacy loss $G_1\epsilon = 11\epsilon$ that we have prior to the exectution of \textsc{Clip-User}. The $\epsilon$-axis is shown on a log-scale, here. From the plots, we see a clear improvement in the privacy loss under composition, for most values of $\epsilon$ considered. Figure \ref{fig:itms2} plots  the worst-case error  ${\overline{E}}_\epsilon$ obtained after the execution of the pseudo-user creation-based mechanism (which in turn is run after the execution of \textsc{Clip-User}) against the original worst-case error $E = E_\epsilon$ prior to the execution of \textsc{Clip-User}. We observe that there is little to no improvement in the worst-case error, across all grids. One reason for this phenomenon could be the fact that the grid contributing to the original worst-case error $E_\epsilon$ contains very few contributions from any user; an execution of the pseudo-user creation-based algorithm may hence not afford much improvement in the worst-case error, since the optimization over $m\in [\Gamma_{g,\star}:\Gamma_g^\star]$ in \eqref{eq:tkde1} may be over very few values.}

\begin{figure}
	\centering
	\includegraphics[width = \linewidth]{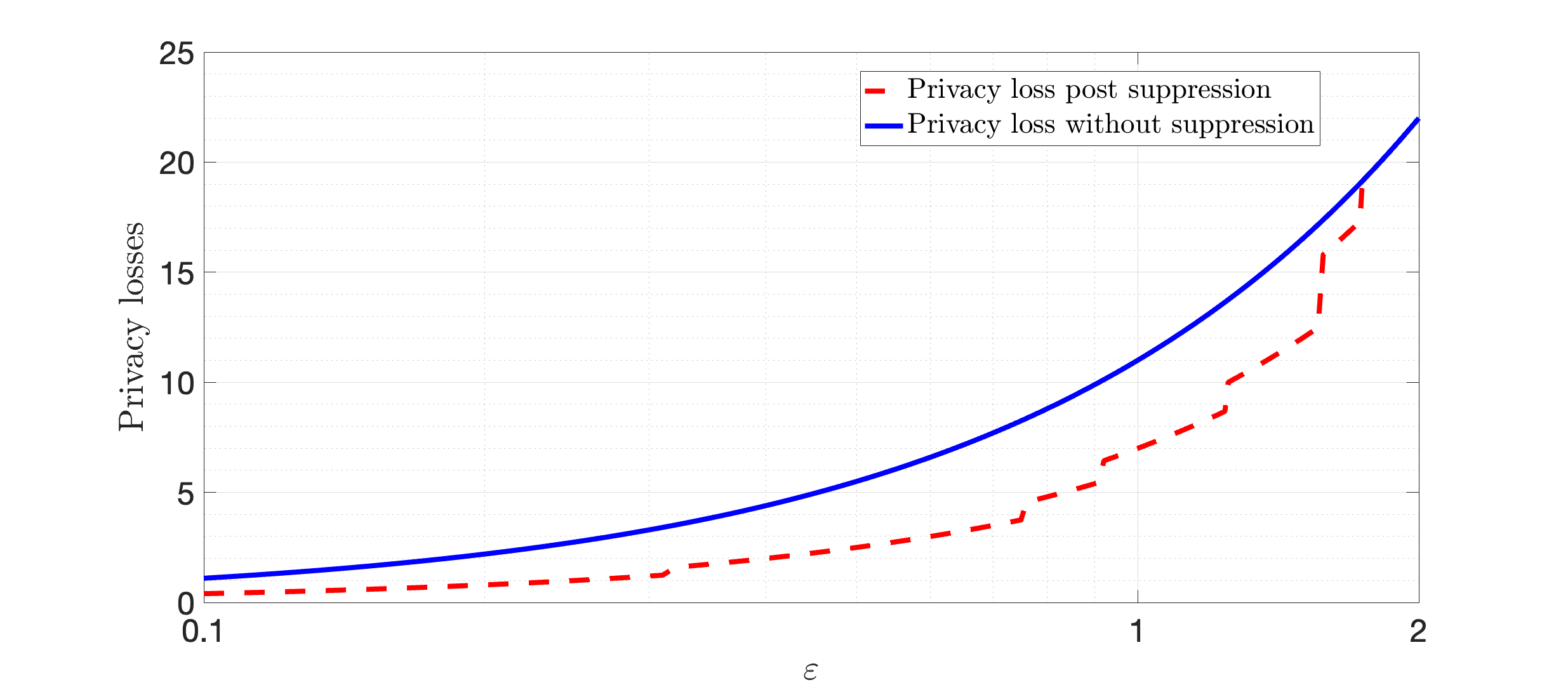}
	\caption{Plot of privacy loss under composition ${P}_\epsilon$ after execution of \textsc{Clip-User} on the real-world ITMS dataset, against the original privacy loss $G_1\epsilon = 11\epsilon$.}
	\label{fig:itms1}
\end{figure}

\begin{figure}
	\centering
	\includegraphics[width = \linewidth]{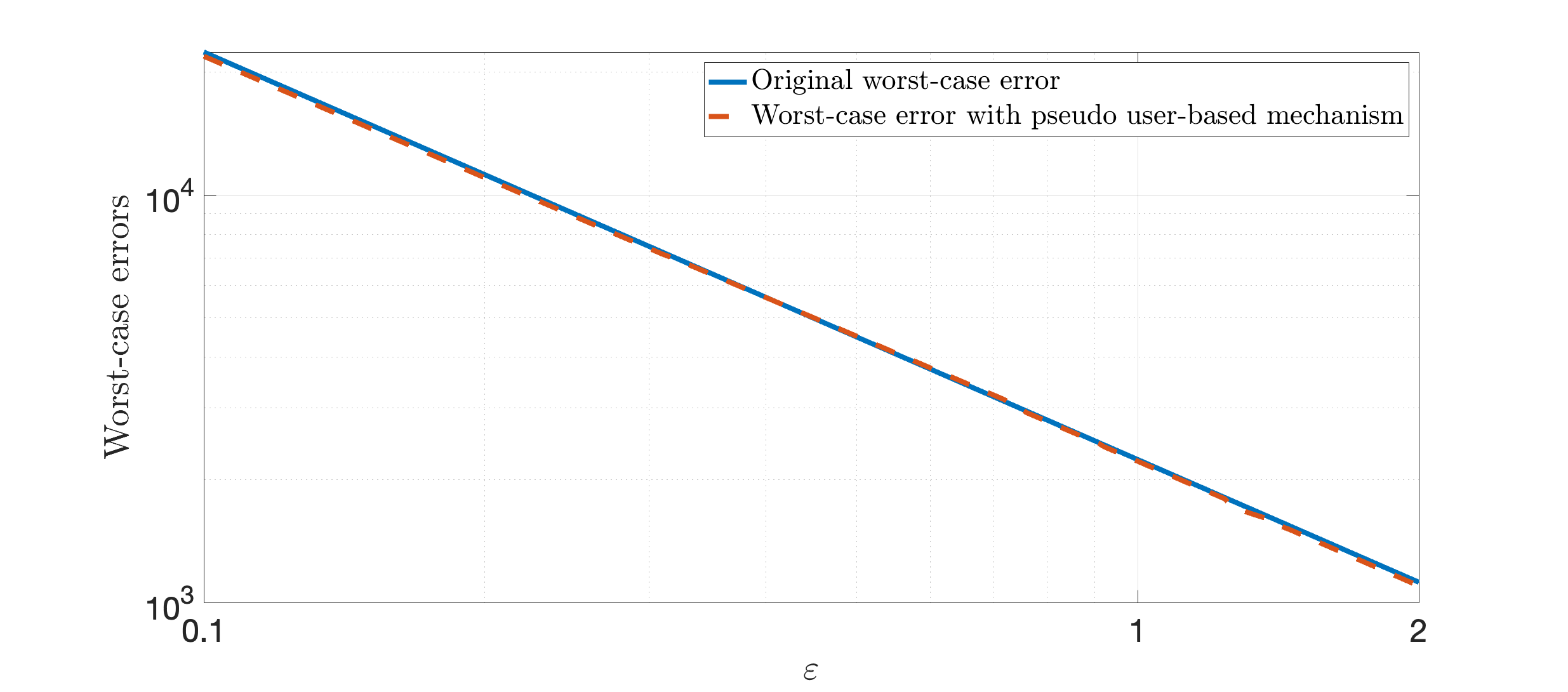}
	\caption{Plot of worst-case error $\overline{E}_\epsilon$ after execution of \textsc{Clip-User} and the implementation of the pseudo-user creation-based clipping strategy on the real-world ITMS dataset, against the original worst-case error $E_\epsilon$.}
	\label{fig:itms2}
\end{figure}

For the synthetic datasets, Figures \ref{fig:K3}--\ref{fig:K9} show plots of the variation of the estimate $\widehat{P}_\epsilon$ of the expected privacy loss  against the original privacy loss  $G\epsilon = G_1\epsilon$ prior to the execution of \textsc{Clip-User}. The $\epsilon$-axis is shown on a log-scale, here. From the plots, it is clear that for a fixed $q\in [0,1]$, increasing $\gamma$ improves the privacy loss degradation. Intuitively, a large value of $\gamma$ leads to a large sensitivity of the unclipped mean and variance (and therefore a large worst-case error $E$); therefore, it is reasonable to expect many stages of \textsc{Clip-User} to execute before the algorithm halts, in this case.

\begin{figure}
	\centering
	\includegraphics[width = \linewidth]{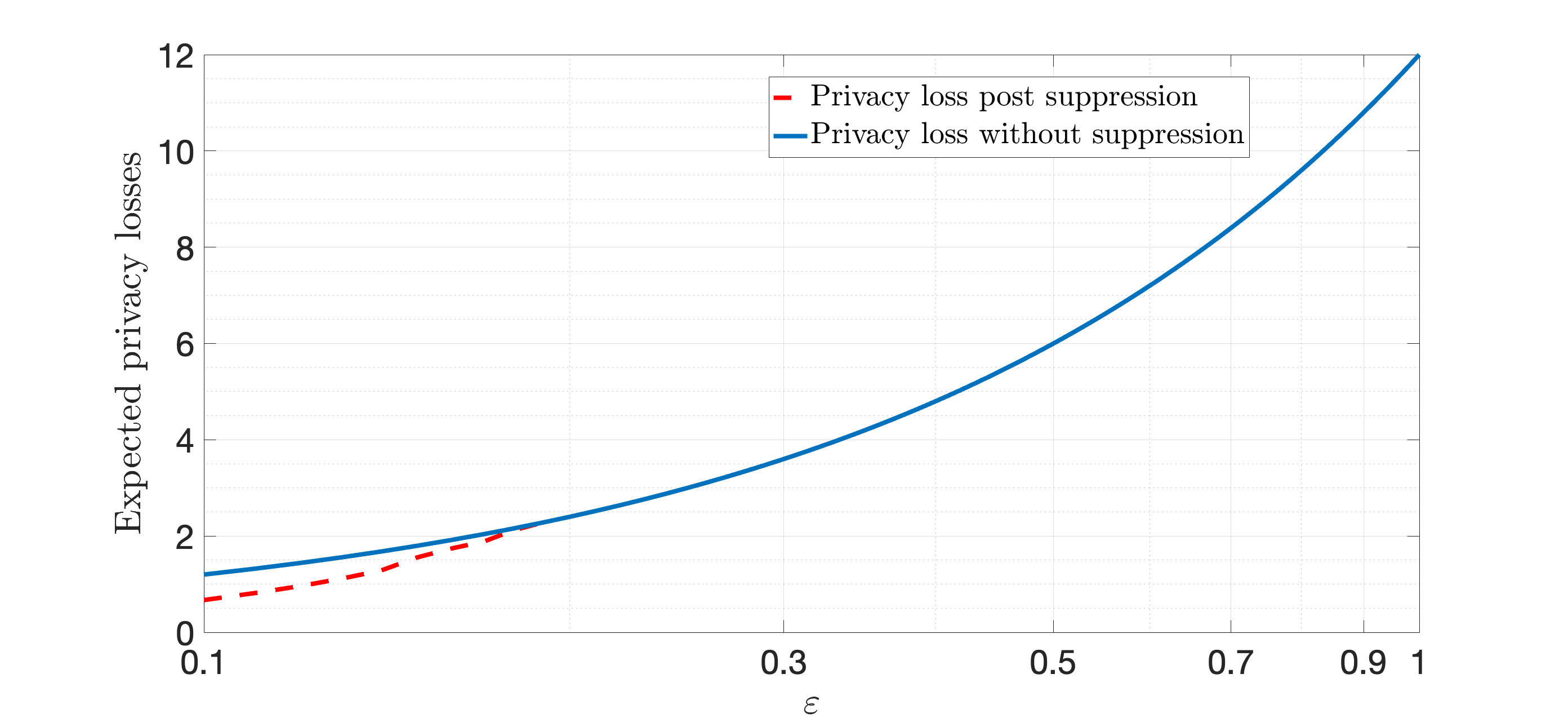}
	\caption{Plot of estimate $\widehat{P}_\epsilon$ of expected privacy loss, after execution of \textsc{Clip-User} on a synthetic dataset,  against the original privacy loss  $G\epsilon$. Here, $\gamma = 3$ and $q = 0.01$.}
	\label{fig:K3}
\end{figure}

\begin{figure}
	\centering
	\includegraphics[width = \linewidth]{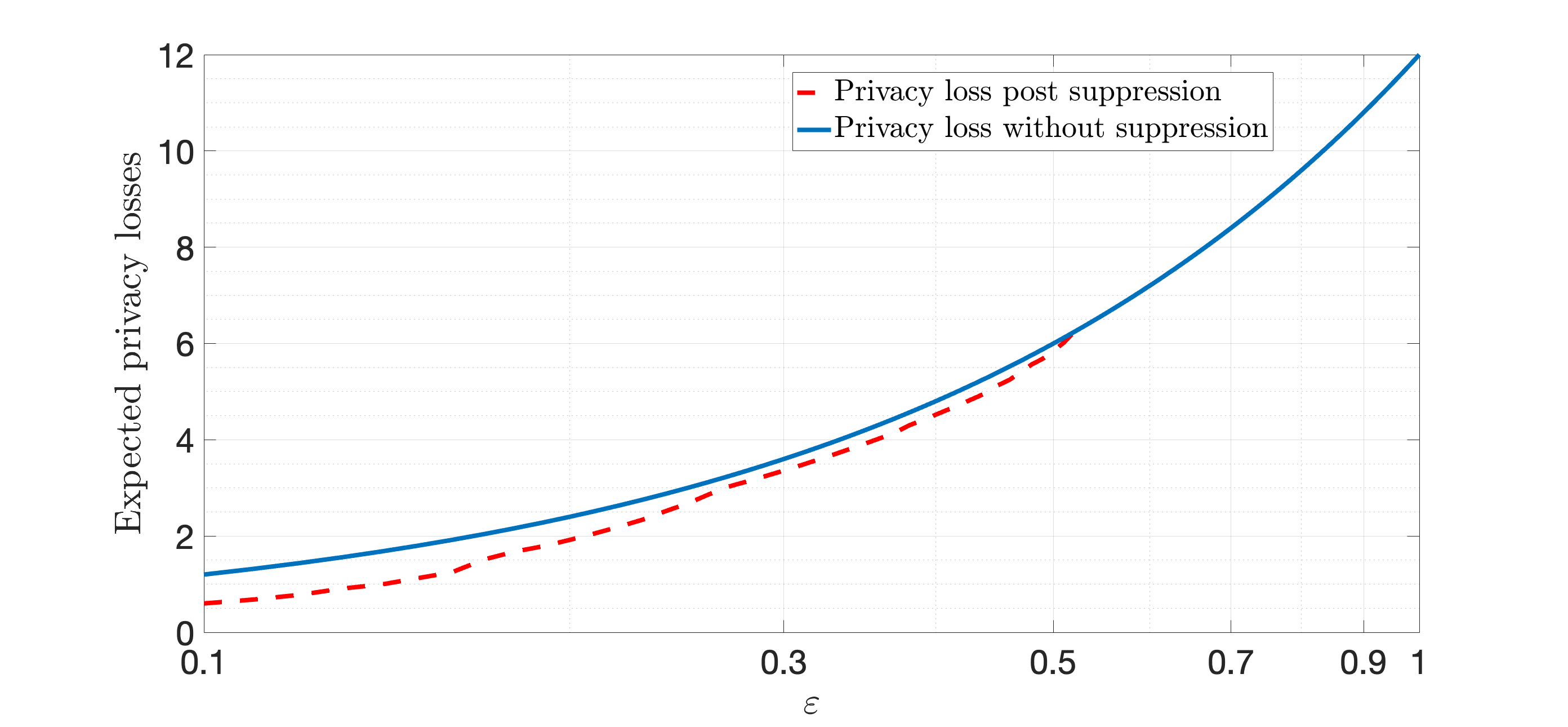}
	\caption{Plot of estimate $\widehat{P}_\epsilon$ of expected privacy loss, after execution of \textsc{Clip-User} on a synthetic dataset, against the original privacy loss  $G\epsilon$. Here, $\gamma = 6$ and $q = 0.01$.}
	\label{fig:K6}
\end{figure}

\begin{figure}
	\centering
	\includegraphics[width = \linewidth]{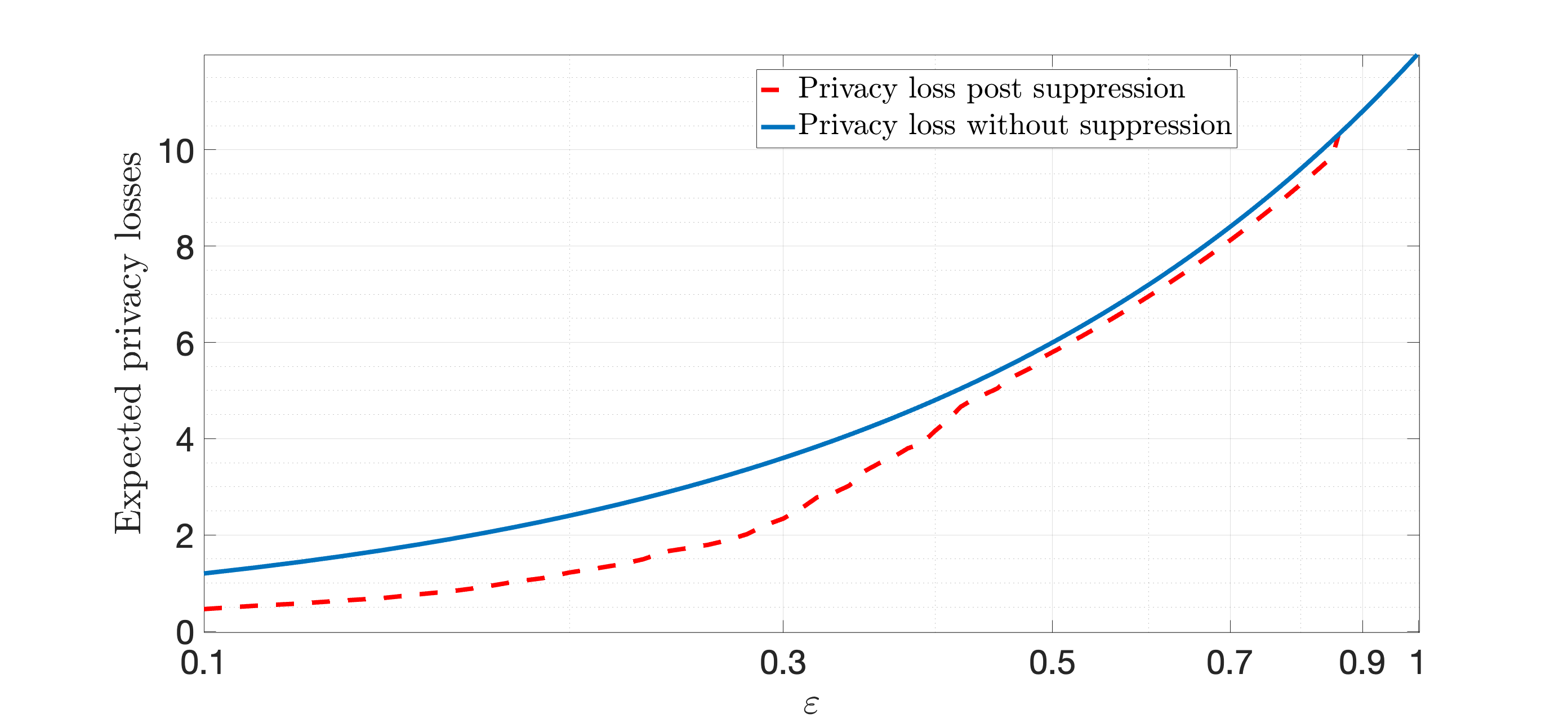}
	\caption{Plot of estimate $\widehat{P}_\epsilon$ of expected privacy loss, after execution of \textsc{Clip-User} on a synthetic dataset, against the original privacy loss  $G\epsilon$.  Here, $\gamma = 9$ and $q = 0.01$.}
	\label{fig:K9}
\end{figure}

Figures \ref{fig:err3}--\ref{fig:err9} for the synthetic datasets show plots of the variation of the estimate of the expected worst-case error across grids $\widehat{\overline{E}}_\epsilon$ against the original worst-case error $E = E_\epsilon$ prior to the execution of \textsc{Clip-User}. Both the $\epsilon$- and the error-axes are shown on a log-scale. Again, it is clear that for a fixed $q\in [0,1]$, increasing $\gamma$ leads to a larger difference between the original and the new worst-case errors, following similar intuition {as earlier}. 

\begin{figure}
	\centering
	\includegraphics[width = \linewidth]{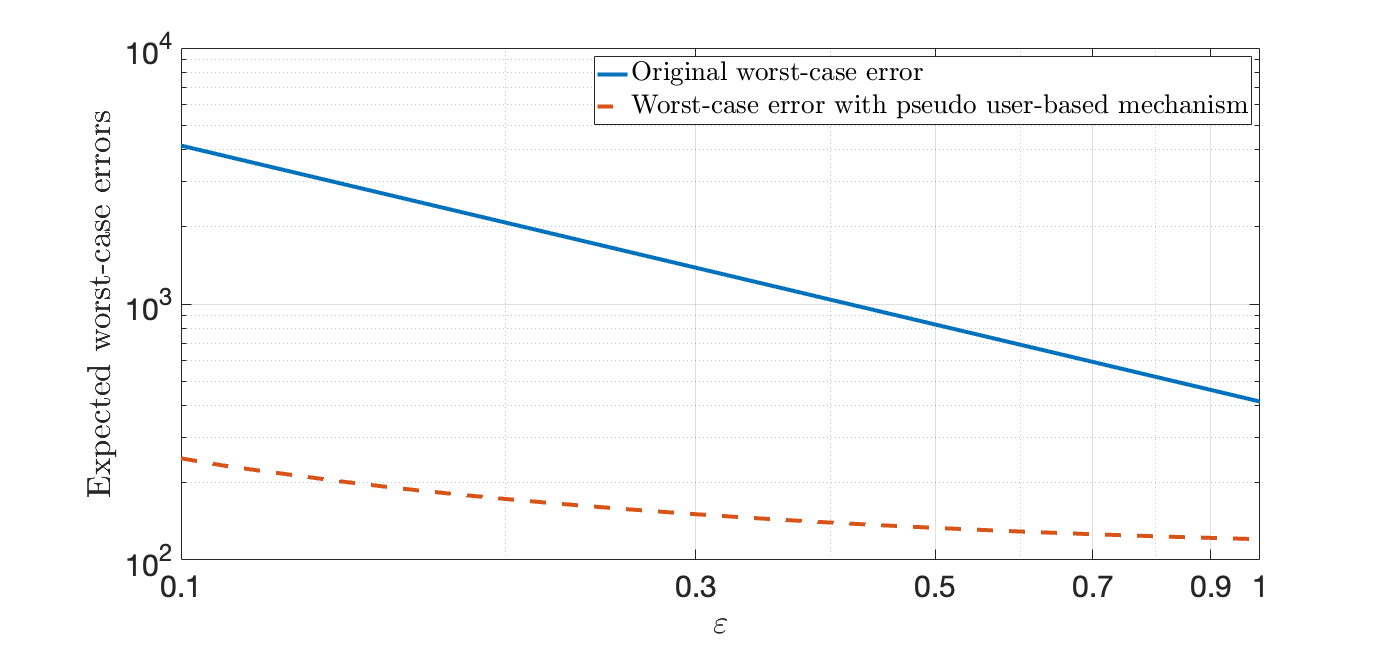}
	\caption{Plot of estimate $\widehat{\overline{E}}_\epsilon$ of the worst-case error across grids, after execution of \textsc{Clip-User} and the implementation of the pseudo-user creation-based clipping strategy on a synthetic dataset, against the original worst-case error $E_\epsilon$. Here, $\gamma = 3$ and $q = 0.01$.}
	\label{fig:err3}
\end{figure}

\begin{figure}
	\centering
	\includegraphics[width = \linewidth]{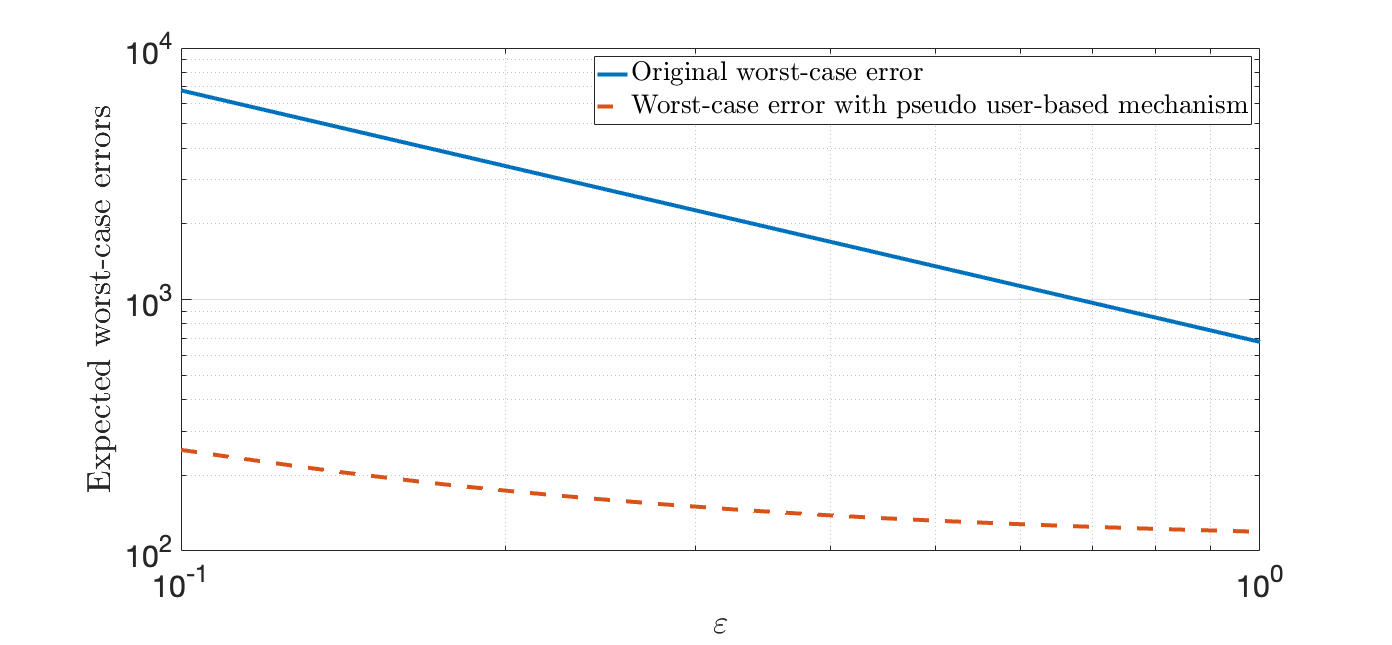}
	\caption{Plot of estimate $\widehat{\overline{E}}_\epsilon$ of the worst-case error across grids, after execution of \textsc{Clip-User} and the implementation of the pseudo-user creation-based clipping strategy on a synthetic dataset, against the original worst-case error $E_\epsilon$. Here, $\gamma = 6$ and $q = 0.01$.}
	\label{fig:err6}
\end{figure}

\begin{figure}
	\centering
	\includegraphics[width = \linewidth]{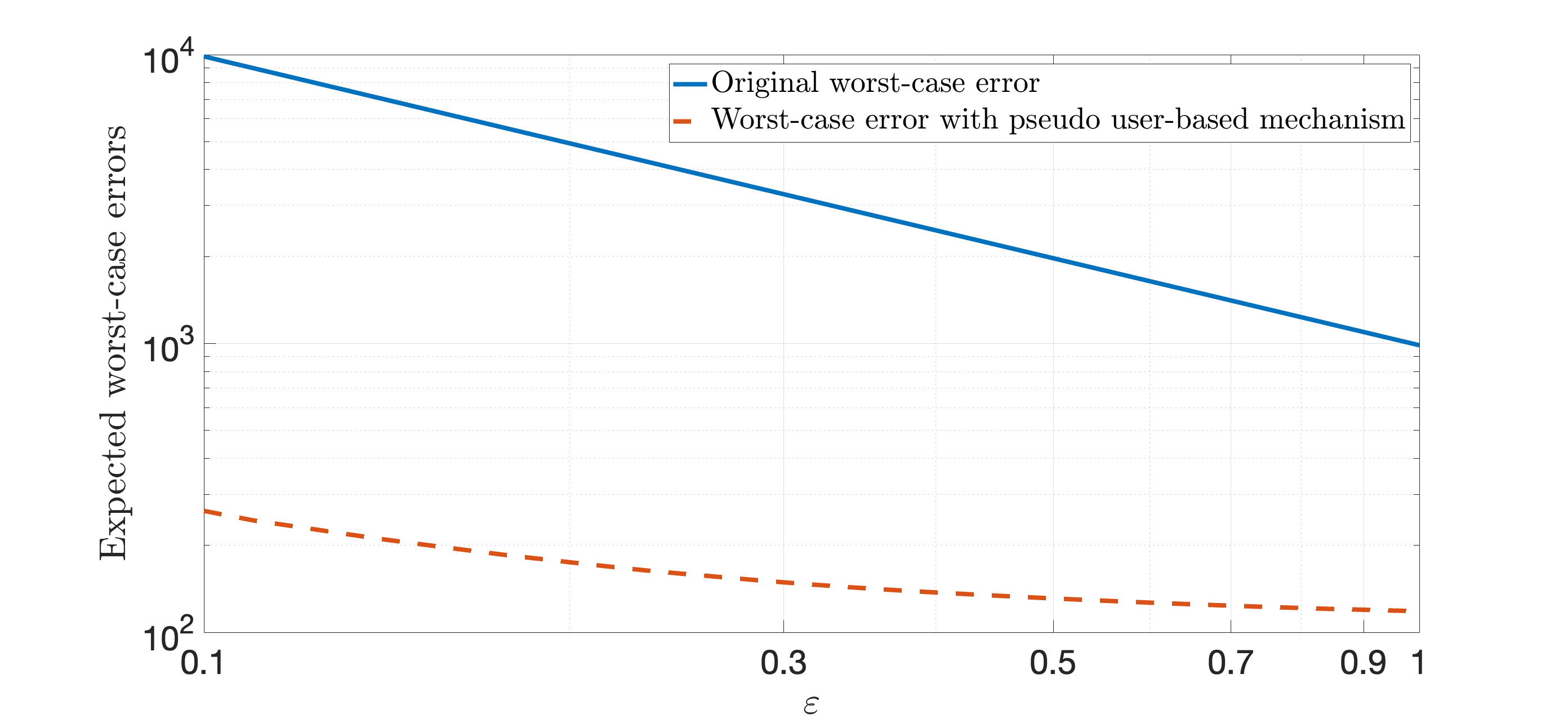}
	\caption{Plot of estimate $\widehat{\overline{E}}_\epsilon$ of the worst-case error across grids, after execution of \textsc{Clip-User} and the implementation of the pseudo-user creation-based clipping strategy on a synthetic dataset, against the original worst-case error $E_\epsilon$. Here, $\gamma = 9$ and $q = 0.01$.}
	\label{fig:err9}
\end{figure}


%% file: conclusion.tex
\label{sec:conclusion}
In this paper, we {considered the worst-case error incurred in the differentially private release of sample means and variances of disjoint subsets of a dataset, in the setting where each user can contribute more than one sample to the dataset, with different users potentially contributing different numbers of samples}. We proposed an algorithm for improving the privacy loss degradation under the composition of user-level (pure) differentially private mechanisms that act on disjoint subsets of a dataset, in such a manner as to maintain the \emph{worst-case} error in estimation over all such subsets. The basic idea behind our algorithm was the suppression of user contributions in selected subsets to improve the privacy loss degradation, while not increasing the worst-case estimation error. 
Key components of the design of our algorithm were our explicit, analytical computations of the sensitivity of the sample variance function and the worst-case bias errors in estimation of the variance arising from clipping selected contributions of users. We also presented a simple extension of a pseudo-user creation-based algorithm, drawing ideas from \cite{dp_preprint}, for reducing the worst-case error across subsets, when the number of users contributing to any subset is fixed. Finally, we evaluated the performance of our algorithms numerically on real-world and synthetically generated datasets, showing discernible improvements in the privacy loss under composition for fixed worst-case estimation error and in the worst-case error across grids, for fixed numbers of user contributions.

{An interesting line of future research would be to extend the analysis here to the setting of (approximate) $(\epsilon,\delta)$-DP; crucially, such analysis will involve making use of advanced composition theorems (see, e.g., \cite[Sec. 3.5.2]{dworkroth}, \cite{gaussiandp}) for privacy accounting, instead of the basic composition theorem. Further, our analysis of the worst-case (bias) errors in clipping-based user-level differentially private estimators can be extended to other statistics of interest}, thereby leading to natural algorithms (along the lines of \textsc{Clip-User}) for improving privacy loss under composition, for fixed worst-case error. Another research direction would be quantifying the tradeoffs between privacy loss and estimation error in learning-based or combinatorial optimization-based differentially private inference tasks on geospatial datasets.

%% file: app-lem-app1.tex
\label{app:lem-app-1}
In this section, we prove Lemma \ref{lem:app1}.
\begin{proof}
	First, we write 
	\begin{align*}
		\Delta_{\textsf{Var}} &=  \max_{\mathcal{D} \sim \mathcal{D}'} \left( \textsf{Var}(\mathcal{D}) - \textsf{Var}(\mathcal{D}')\right)\\
		&= \max_{\mathcal{D}} \left(\textsf{Var}(\mathcal{D}) - \min_{\mathcal{D}'\sim \mathcal{D}}\textsf{Var}(\mathcal{D}')\right).	
	\end{align*}
	Now, for a fixed dataset $\mathcal{D}$, consider $\textsf{Var}(\mathcal{D}')$, for $\mathcal{D}'\sim \mathcal{D}$. Let $\tilde{X}\sim \text{Unif}(\tilde{A}\cup A^c)$ denote a uniformly distributed random variable that takes values in the set $\left\{\tilde{S}_\ell^{(j)}\right\}$. Then,
	\begin{align*}
		\textsf{Var}(\mathcal{D}')&= \E\left[(\tilde{X}-\tilde{\nu})^2\right]\\
		&= \E\left[(\tilde{X}-\tilde{\nu})^2\mid \tilde{X}\in \tilde{A}\right] \Pr[\tilde{X}\in \tilde{A}]\ +  \E\left[(\tilde{X}-\tilde{\nu})^2\mid \tilde{X}\in {A}^c\right]  \Pr[\tilde{X}\in A^c]\\
		&\stackrel{(a)}{=} \E\left[(\tilde{X}-\tilde{\nu})^2\mid \tilde{X}\in \tilde{A}\right]\cdot \left(\frac{m_k}{\sum_\ell m_\ell}\right)\ +  \E\left[(\tilde{X}-\tilde{\nu})^2\mid \tilde{X}\in {A}^c\right]\cdot \left(1- \frac{m_k}{\sum_\ell m_\ell}\right).
	\end{align*}
	Now, consider the term $\E\left[(\tilde{X}-\tilde{\nu})^2\mid \tilde{X}\in \tilde{A}\right]$ above. We can write
	\begin{align*}
		&\E\left[(\tilde{X}-\tilde{\nu})^2\mid \tilde{X}\in \tilde{A}\right]\\ &= \E\left[(\tilde{X}-{\nu}(\tilde{A})+{\nu}(\tilde{A})-\tilde{\nu})^2\mid \tilde{X}\in \tilde{A}\right]\\
		&= \E\left[(\tilde{X}-{\nu}(\tilde{A}))^2\mid \tilde{X}\in \tilde{A}\right]+(\tilde{\nu}-{\nu}(\tilde{A}))^2\ +  2\ ({\nu}(\tilde{A})-\tilde{\nu})\cdot \E\left[(\tilde{X}-{\nu}(\tilde{A}))\mid \tilde{X}\in \tilde{A}\right].
	\end{align*}
	Clearly, since $\tilde{X}\sim \text{Unif}(\mathcal{D}')$, we have that conditioned on the event $\{\tilde{X}\in \tilde{A}\}$, we have that $\tilde{X}$ is uniform on $\tilde{A}$. Therefore, we obtain that $\E\left[(\tilde{X}-{\nu}(\tilde{A}))\mid \tilde{X}\in \tilde{A}\right] = 0$, implying that  
	\begin{equation}
		\label{eq:app1}
		\E\left[(\tilde{X}-\tilde{\nu})^2\mid \tilde{X}\in \tilde{A}\right] = \E\left[(\tilde{X}-{\nu}(\tilde{A}))^2\mid \tilde{X}\in \tilde{A}\right] + (\tilde{\nu}-{\nu}(\tilde{A}))^2.
	\end{equation}
	By similar arguments, we obtain that 
	\begin{equation}
		\label{eq:app2}
		\E\left[(\tilde{X}-\tilde{\nu})^2\mid \tilde{X}\in {A}^c\right] = \E\left[(\tilde{X}-{\nu}(A^c))^2\mid \tilde{X}\in {A}^c\right]+(\tilde{\nu}-{\nu}({A}^c))^2.
	\end{equation}
	Substituting \eqref{eq:app1} and \eqref{eq:app2} into equality (a) above, we get that
	\begin{align*}
		\textsf{Var}(\mathcal{D}')&=  \left(\E\left[(\tilde{X}-{\nu}(\tilde{A}))^2\mid \tilde{X}\in \tilde{A}\right] + (\tilde{\nu}-{\nu}(\tilde{A}))^2\right)\cdot \left(\frac{m_k}{\sum_\ell m_\ell}\right)\ + \\
		&\ \ \ \ \  \ \ \ \ \ \  \  \left(\E\left[(\tilde{X}-{\nu}(A^c))^2\mid \tilde{X}\in {A}^c\right]+(\tilde{\nu}-{\nu}({A}^c))^2\right)\cdot \left(1- \frac{m_k}{\sum_\ell m_\ell}\right).
	\end{align*}
	
	Now, observe that all the terms in equality (a) above are non-negative, and hence $\textsf{Var}(\mathcal{D}')$ is minimized by setting $\tilde{\nu} = {\nu}(\tilde{A}) = {\nu}({A}^c) = \tilde{S}_k^{(j)}$, for all $j\in [m_k]$. 
\end{proof}

%% file: app-lem-app2.tex
\label{app:lem-app-2}
In this section, we shall prove Lemma \ref{lem:app2}.

\begin{proof}
	Recall that
	\[
	\Delta_{\textsf{Var}} = \textsf{Var}(\mathcal{D}_1) - \textsf{Var}(\mathcal{D}_2),
	\]
	with $\mathcal{D}_2$ chosen as in the discussion preceding this lemma. Thus, for random variables $\tilde{X}\sim \text{Unif}(\tilde{A}\cup A^c)$ and $X\sim \text{Unif}(A\cup A^c)$, we have
	\begin{align*}
		\Delta_{\textsf{Var}}
		&= \max_{\mathcal{D}} \left[ \textsf{Var}(\mathcal{D}) - \E\left[(\tilde{X}-{\nu}(A^c))^2\mid \tilde{X}\in {A}^c\right]\cdot \left(1- \frac{m_k}{\sum_\ell m_\ell}\right)\right]\\
		&=  \max_{\mathcal{D}} \left[\E\left[(X-\nu)^2\right] - \E\left[(\tilde{X}-{\nu}(A^c))^2\mid \tilde{X}\in {A}^c\right]\cdot \left(1- \frac{m_k}{\sum_\ell m_\ell}\right)\right]\\
		&\stackrel{(b)}{=} \max_{\mathcal{D}} \left[\E\left[(X-\nu)^2\right] - \E\left[({X}-{\nu}(A^c))^2\mid {X}\in {A}^c\right]\cdot \left(1- \frac{m_k}{\sum_\ell m_\ell}\right)\right],
	\end{align*}
	where equality (b) follows from the fact that the distribution of $\tilde{X}$ conditioned on the event $\{\tilde{X}\in A^c\}$ is identical to that of $X$ conditioned on the event $\{X\in A^c\}$. Hence,
	\begin{align*}
		&\Delta_{\textsf{Var}}= \max_{\mathcal{D}} \Bigg[\E\left[(X-\nu)^2\mid X\in A\right]\Pr[X\in A] + \E\left[(X-\nu)^2\mid X\in A^c\right]\Pr[X\in A^c] \\
		& \ \ \ \ \  \ \ \ \ \ \ \ \ \ \ \  \ \ \ \ \ \ \  \ \ \ \ \ \ \ \ \ \ \  \ \ \ \ \ \ \  \ \ \ \ \ \ \ \ \ \ \  \ \ - \E\left[(\tilde{X}-{\nu}(A^c))^2\mid \tilde{X}\in {A}^c\right]\cdot \left(1- \frac{m_k}{\sum_\ell m_\ell}\right)\Bigg].
	\end{align*}
	Now, observe that by arguments as in the proof of Lemma \ref{lem:app1}, 
	\begin{align*}
		\E\left[(X-\nu)^2\mid X\in A^c\right]
		&= \E\left[({X}-{\nu}(A^c))^2\mid {X}\in {A}^c\right]+ (\nu-\nu(A^c))^2.
	\end{align*}
	Now, since $\Pr[X\in A^c] = \left(1- \frac{m_k}{\sum_\ell m_\ell}\right)$, we have by arguments made earlier, that 
	\begin{align}
		\Delta_{\textsf{Var}}
		&= \max_{\mathcal{D}} \bigg[\E\left[(X-\nu)^2\mid X\in A\right]\Pr[X\in A]+  (\nu-\nu(A^c))^2 \Pr[X\in A^c]\bigg] \label{eq:app3}\\
		&= \max_{\mathcal{D}:\ {S}_\ell^{(j)} = \nu(A^c), \forall {S}_\ell^{(j)}\in A^c} \textsf{Var}(\mathcal{D}), \notag
	\end{align}
	thereby proving the lemma.
\end{proof}

%% file: app-varsens.tex
\label{app:versens}
In this section, we shall prove Proposition \ref{prop:varsens}. 


\begin{proof}[Proof of Proposition \ref{prop:varsens}]
	Recall from Lemma \ref{lem:app2} that
	\[
	\Delta_{\textsf{Var}} = \max_{\mathcal{D}:\ {S}_\ell^{(j)} = \nu(A^c), \forall {S}_\ell^{(j)}\in A^c} \textsf{Var}(\mathcal{D}).
	\]
	From the discussion following Lemma \ref{lem:app2}, consider the case when the maximum over $\nu(A^c)$ above is attained at $\nu(A^c) = 0$. The proof for the case when $\nu(A^c) = U$ follows along similar lines, and is hence omitted. In this case,
	\[
	\Delta_{\textsf{Var}} = \max_{{S}_\ell^{(j)} \in A:\ {S}_\ell^{(j)} = 0, \forall {S}_\ell^{(j)}\in A^c} \textsf{Var}(\mathcal{D}).
	\]
	In this setting, $\nu \leq \frac{Um_k}{\sum_\ell m_\ell}$. Two possible situations arise: (i) when $\sum_\ell m_\ell>2m^\star$, and (ii) when $\sum_\ell m_\ell\leq 2m^\star$. Consider the first situation. In this case, observe that $\nu \leq \frac{Um^\star}{\sum_\ell m_\ell}< U/2$. Further, from the Bhatia-Davis inequality \cite{davisbhatia}, we have $\textsf{Var}(\mathcal{D})\leq \nu(U-\nu) =: b(\nu)$. Hence, for the range of $\nu$ values of interest, we have that $b(\nu)$ is strictly increasing in $\nu$. Hence,
	\begin{align*}
		\Delta_{\textsf{Var}}&\leq \max_{\mathcal{D}:\ {S}_\ell^{(j)} = 0, \forall {S}_\ell^{(j)}\in A^c}\nu(U-\nu)\leq \frac{U^2\ m^\star(\sum_\ell  m_\ell - m^\star)}{\left(\sum_\ell m_\ell\right)^2},
	\end{align*}
with the inequalities above being achieved with equality when $S_k^{(j)} = U$, for all $j\in [m_k]$, and when $m_k = m^\star$. 
Next, consider the situation when {$\sum_\ell m_\ell\leq 2m^\star$}, and suppose that $\sum_\ell m_\ell$ is even. In this case, we have that $|A|\geq |A^c|$. For this setting, first note that 
\[
\Delta_{\textsf{Var}}\leq \max_{\mathcal{D}\in \mathsf{D}} \textsf{Var}(\mathcal{D}) = \max \textsf{Var}(W) = \frac{U^2}{4},
\]
for $W\sim \text{Unif}(A\cup A^c)$. To see why the above bound holds, note that for any bounded random variable $Y\in [0,U]$, we have that $$\textsf{Var}(Y) = \textsf{Var}(Y - U/2) \leq U^2/4.$$ Furthermore, equality above is attained when all samples in $A^c$ take the value $0$ (which is in line with the case of interest where $\nu(A^c) = 0$) and $\frac{|A|-|A^c|}{2}$ samples in $A$ take the value $0$ and the remaining samples take the value $U$. This then results in exactly $\frac{\sum_\ell m_\ell}{2}$ samples being $0$ and an equal number of samples being $U$, resulting $\Delta_{\textsf{Var}} = U^2/4$.

Next, consider the case when $\sum_\ell m_\ell$ is odd. In this setting, it is not possible to ensure that equal number of samples (from $A\cup A^c$) are at $0$ and $U$, thereby implying that the true value of $\textsf{Var}(\mathcal{D})$, with $S_j^{(\ell)} = 0$, for all $S_j^{(\ell)} \in A^c$, in this case is smaller than $U^2/4$. We claim that in the case when the total number, $\sum_\ell m_\ell$, of samples is odd, the variance of a bounded random variable $Y\in [0,U]$ that takes values in $\left\{S_j^{(\ell)}\right\}$ obeys
\begin{equation}
	\label{eq:varodd}
\textsf{Var}(Y) \leq \frac{U^2}{4}\cdot \left(1-\frac{1}{(\sum_\ell m_\ell)^2}\right);
\end{equation}
furthermore, this bound is achieved when $\left \lceil \frac{\sum_\ell m_\ell}{2}\right \rceil$ samples take the value $0$ and $\left \lfloor \frac{\sum_\ell m_\ell}{2}\right \rfloor$ samples take the value $U$. Modulo this claim, observe that in the case where $|A|\geq |A^c|$, the upper bound in \eqref{eq:varodd} is achievable when $S_j^{(\ell)} = 0$, for all $S_j^{(\ell)} \in A^c$, by placing $\left \lceil \frac{|A|-|A^c|}{2}\right \rceil$ samples from $A$ at the value $0$ and the remaining samples at $U$.

We now prove the above claim. To this end, we first show that any sample distribution $\left\{S_j^{(\ell)}\right\}$ that maximizes the variance above must be such that $S_j^{(\ell)}\in \{0,U\}$, for all $\ell,j$. For ease of reading, we let the samples $\left\{S_j^{(\ell)}\right\}$ be written as the collection $\{x_1,\ldots,x_n\}$, where $n = \sum_\ell m_\ell$. Now, we write
\begin{equation}
	\label{eq:variter}
\max_\mathcal{D} \textsf{Var}(Y) = \max_{x_1}\max_{x_2}\ldots\max_{x_n} \textsf{Var}(Y).
\end{equation}
Note that for fixed values of $x_1,x_2,\ldots,x_{n-1}$, the variance above is maximized when $x_n\in \{0,U\}$. To see why, let $\nu_{\sim n}$ denote the sample mean of the samples $x_1,\ldots,x_{n-1}$ and let $Y_{\sim n}$ denote the random variable that is uniformly distributed over the samples $x_1,x_2,\ldots,x_{n-1}$. By arguments as earlier, note that
\begin{align*}
	\textsf{Var}(Y)&= \left(\frac{n-1}{n}\right)\cdot \left(\E\left[(Y-\nu_{\sim n})^2\mid Y\neq x_n\right]+ (\nu-\nu_{\sim n})^2\right)+\frac{1}{n} \left(\nu-x_n\right)^2\\
	&=	\left(\frac{n-1}{n}\right)\cdot \textsf{Var}(Y_{\sim n})+ \frac{n-1}{n^2}\cdot (x_n-\nu_{\sim n})^2.
\end{align*}
Clearly, the above expression is maximized, for fixed $x_1,x_2,\ldots,x_{n-1}$, by $x_n \in \{0,U\}$, depending on the value of $\nu_{\sim n}$. This argument can then be repeated iteratively over all $x_1,\ldots,x_n$, using \eqref{eq:variter}. 

Now, since all the samples in the collection $\{x_1,\ldots,x_n\}$ take a value of either $0$ or $U$, all that remains is a maximization of $\textsf{Var}(Y)$, given this constraint. Let $k$ denote the number of samples taking the value $0$ and let $n-k$ be the number of samples taking the value $U$. In this case, $\nu = \frac{(n-k)U}{n}$. Then,
\begin{align*}
	\frac{\textsf{Var}(Y)}{U^2} &= \frac{k}{n}\cdot \left(\frac{n-k}{n}\right)^2+\frac{n-k}{n}\cdot \left(\frac{k}{n}\right)^2\\
	&= \frac{k(n-k)}{n^2}.
\end{align*}
Clearly, when $n$ is odd, the above expression is maximized when $\left \lceil \frac{n}{2}\right \rceil$ values are $0$ and the remaining $\left \lfloor \frac{n}{2}\right \rfloor$ values are $U$, proving our earlier claim.
\end{proof}

%% file: app-array-av.tex
\label{app:array-av}
In this section, we consider a special class of clipping strategies obtained by setting $\Gamma_\ell = \min\{m,m_\ell\}$, for some fixed $m\in [m_\star:m^\star]$. Clearly, here, we have $\Gamma^\star = m$ and $\Gamma_{\star}:= \min_{\ell\in \mathcal{L}} \Gamma_{\ell} = m_\star$. Such a clipping strategy arises naturally in the design of user-level differentially private mechanisms based on the creation of pseudo-users \cite{tyagi, dp_preprint}. We show that for choices of $m$ of interest, the sensitivities of the clipped sample mean and variance are at most those of their unclipped counterparts. In particular, for the sample mean, the following lemma was shown in \cite{dp_preprint}:


\begin{lemma}[Lemma III.1 in \cite{dp_preprint}]
	For any $m\leq m^\star$, we have that $\Delta_\mu \geq {\Delta}_{\mu_\text{clip}}$.
\end{lemma}
We now proceed to state and prove an analogous lemma that compares the sensitivities of $\textsf{Var}_\text{clip}$ and $\textsf{Var}$. Before we proceed, observe that it is natural to restrict attention to those values of $m\in [m_\star:m^\star]$ that minimize the sensitivity of the clipped variance in \ref{lem:vararrsens}. We first show that there exists a minimizer $m\in [m_\star:m^\star]$ that takes its value in the set $\{m_\ell\}_{\ell\in \mathcal{L}}$. Let $\Delta_{\textsf{Var}_{\text{clip}}}(m) = \Delta_{\textsf{Var}_{\text{clip}}}$, for a fixed $m$. To achieve this objective, we need the following helper lemma. For ease of exposition, we assume that $m_\star = m_1\leq m_2\leq\ldots\leq m_L = m^\star$. We also assume throughout that $L\geq 2$.

\begin{lemma}
	\label{lem:sensconcave}
	$\Delta_{\textsf{Var}_{\text{clip}}}(m)$ is concave in $m$, for $m\in [m_t,m_{t+1}]$, for any $t\in [L-3]$, when $L\geq 3$.
\end{lemma}
\begin{proof}
	Fix an integer $t\in [L-1]$, for $L\geq 3$. 
	Let $\alpha_1(m):= \frac{ \Gamma_\ell^\star(\sum_\ell \Gamma_\ell - \Gamma_\ell^\star)}{\left(\sum_\ell \Gamma_\ell\right)^2}$, $\alpha_2(m):= 
	\frac{1}{4}$, and $\alpha_3(m):=
	\frac{1}{4}\cdot\left(1-\frac{1}{(\sum_\ell \Gamma_\ell)^2}\right)$.
	
Now, consider the setting where $t\leq L-3$. In this case, observe that 
	\begin{align*}\sum_\ell \Gamma_\ell &= \sum_\ell \min\{m,m_\ell\}\\
		&= \sum_{\ell=1}^t m_t + (L-t)m> 2m,
	\end{align*}
by our choice of $t$. This implies that for such values of $t$, we have $\Delta_{\textsf{Var}_{\text{clip}}}(m) = U^2\cdot\alpha_1(m)$, for all $m\in [m_t, m_{t+1}]$. Now, observe that we can write
\begin{align*}
{\alpha}_1(m) &= \frac{m(c_1+m(c_2-1))}{(c_1+c_2m)^2}\\
&= \frac{m}{c_1+c_2m} - \frac{m^2}{(c_1+c_2m)^2}=: a_1(m)+b_1(m),
\end{align*}
for constants $c_1,c_2>0$ such that $c_1+c_2m = \sum_\ell \Gamma_{\ell}$. By direct computation, it is possible to show that 
\[
\frac{\d^2 a_1}{\d m^2} = \frac{-2c_1c_2}{(c_1+c_2m)^3}<0
\]
and
\[
\frac{\d^2 b_1}{\d m^2} = -2c_1\cdot \left(\frac{c_1+m(c_2-3)}{(c_1+c_2m)^4}\right)\leq 0,
\]
since $c_1+c_2m = \sum_\ell \Gamma_\ell \geq 3m$, by our choice of $t$. Hence, for this case, we obtain that $\Delta_{\textsf{Var}_{\text{clip}}}(m)$ is concave in $m$.
\end{proof}
We are now ready to show that there exists a minimizer of the sensitivity $\Delta_{\textsf{Var}_{\text{clip}}}$ that takes its value in $\{m_\ell\}_{\ell \in \mathcal{L}}$.
\begin{lemma}
	There exists
	$
	m\in \argmin_{m\in [m_\star,m^\star]} \Delta_{\textsf{Var}_{\text{clip}}}(m),
	$
	such that $m\in \{m_\ell\}_{\ell\in \mathcal{L}}$.
\end{lemma}
\begin{proof}

Suppose that $m\in [m_t,m_{t+1}]$, for some $t\in [L-1]$. We now argue that the value of $\Delta_{\textsf{Var}_{\text{clip}}}(m)$ cannot increase by setting $m$ to $\argmin_{m\in \{m_t,m_{t+1}\}} \Delta_{\textsf{Var}_{\text{clip}}}(m)$. Indeed, note that if $t\in [L-3]$, by the concavity of $\Delta_{\textsf{Var}_{\text{clip}}}(m)$ from Lemma \ref{lem:sensconcave}, we obtain that a minimizer of $\Delta_{\textsf{Var}_{\text{clip}}}(m)$, for $m\in [m_t,m_{t+1}]$, occurs at a boundary point. 

Now, consider the case when $t = L-2$. In this case, observe that $\sum \Gamma_{\ell} - 2m = \sum_{\ell = 1}^{L-2} m_\ell >0$, if $L>2$, and equals $0$, if $L\leq 2$. Consider the first case when $L>2$. In this setting, we have $\Delta_{\textsf{Var}_{\text{clip}}}(m) = \alpha_1(m)$, for all $m\in [m_{L-2},m_{L-1}]$. It is possible, by direct calculations, to show that when $m\in  [m_{L-2},m_{L-1}]$, we have
\begin{align*}
	\frac{\d \alpha_1}{\d m} = \frac{-2mc_1}{(c_1+c_2m)^3},
\end{align*}
for some constants $c_1,c_2>0$, thereby implying that $\alpha_1$ is decreasing as a function of $m$, in this interval. Therefore, a minimizer of $\Delta_{\textsf{Var}_{\text{clip}}}$ occurs at a boundary point.

Next, consider the case when $L= 2$. In this setting, we have that $\Delta_{\textsf{Var}_{\text{clip}}}(m)$ equals either $\alpha_2(m)$ or $\alpha_3(m)$, for $m\in [m_{L-2},m_{L-1}]$, when $\sum_\ell \Gamma_\ell$ is even or odd, respectively. Since $\alpha_2(m)$ is a constant and $\alpha_3(m)$ can be seen to be increasing in $m$ in this interval, we obtain once again that a minimizer of $\Delta_{\textsf{Var}_{\text{clip}}}$ occurs at a boundary point.


Now, consider the case when $t = L-1$. Observe that in this case, $\sum_\ell \Gamma_{\ell} - 2m = \sum_{\ell=1}^{L-1} m_\ell - m$ is decreasing as $m$ increases from $m_{L-1}$ to $m_L$. Hence, one of three possible cases can occur, each of which is dealt with in turn, below.
\begin{enumerate}
	\item $\sum_\ell \Gamma_{\ell}\leq 2m$, for all $m\in [m_{L-1},m_L]$: Clearly, in this case, we have that $\Delta_{\textsf{Var}_{\text{clip}}}$ equals either $\alpha_2$ or $\alpha_3$, for $m\in [m_{L-2},m_{L-1}]$, when $\sum_\ell \Gamma_\ell$ is even or odd, respectively. Since $\alpha_2$ is a constant and $\alpha_3(m)$ is increasing with $m$ in the interval of interest, we obtain that a minimizer of $\Delta_{\textsf{Var}_{\text{clip}}}$ occurs at a boundary point.
	\item $\sum_\ell \Gamma_{\ell}> 2m$, for all $m\in [m_{L-1},m_L]$: Here, $\Delta_{\textsf{Var}_{\text{clip}}} = \alpha_1$. Furthermore, we have that 
	\[
	\frac{\d \alpha_1}{\d m} = \frac{\sum_{\ell=1}^{L-1} m_\ell}{(m+\sum_{\ell=1}^{L-1} m_\ell)^2}>0,
	\]
	implying that $\Delta_{\textsf{Var}_{\text{clip}}}$ is increasing in the interval of interest, hence showing that its minimizer occurs at a boundary point.
	\item $\sum_\ell \Gamma_{\ell}> 2m$, for $m\in [m_{L-1},\overline{m}]$ and $\sum_\ell \Gamma_{\ell}\leq 2m$, for $m\in (\overline{m},m_L]$, for some $\overline{m}\in [m_{L-1},m_L]$: Observe first that in this setting, we have that when $m = \overline{m}$,
	\[
	\sum_\ell \Gamma_\ell = \overline{m}+\sum_{\ell=1}^{L-1} m_\ell = 2\overline{m},
	\]
	by the definition of $\overline{m}$. In other words, we have $\overline{m} = \sum_{\ell=1}^{L-1}m_\ell$. Furthermore, for $m\in [m_{L-1},\overline{m}]$, we have $\Delta_{\textsf{Var}_{\text{clip}}}(m) = \alpha_1(m)$, while for $m\in (\overline{m},m_L]$, we have $\Delta_{\textsf{Var}_{\text{clip}}}(m)$ equals $\alpha_2(m)$ or $\alpha_3(m)$, respectively, depending on whether $\sum_\ell \Gamma_\ell$ is even or odd. In the case when $\sum_\ell \Gamma_\ell$ is even, it can be verified that $\Delta_{\textsf{Var}_{\text{clip}}}(\overline{m}) = \alpha_2(m) = 1/4$. Thus, using the fact that $\alpha_1(m)$ is increasing in $m$, we obtain that a minimizer of $\Delta_{\textsf{Var}_{\text{clip}}}$ occurs at a boundary point, when $\sum_\ell \Gamma_\ell$ is even.
	
	Next, when $\sum_\ell \Gamma_\ell$ is odd, we have that $\alpha_3(m)$ is increasing in $m$ for the interval of interest; we thus need only verify if for $L\geq 2$, we have 
	\[
	\alpha_3(\overline{m})\geq \alpha_1(m_{L-1}).
	\]
	Indeed, if the above inequality holds, we have that $\Delta_{\textsf{Var}_{\text{clip}}}(\overline{m})$ is minimized at $m = m_{L-1}$, for $m\in [m_{L-1},m_L]$. We can verify that the above inequality indeed holds, by a simple direct computation.
\end{enumerate}
Hence, overall, we have that a minimizer of $\Delta_{\textsf{Var}_{\text{clip}}}({m})$, for $m\in [m_t,m_{t+1}]$ occurs at a boundary point, for all $t\in [L-1]$.
\end{proof}
Now that we have established that it suffices to focus on $m\in \{m_\ell\}_{\ell\in \mathcal{L}}$, we show that $\Delta_{\textsf{Var}_{\text{clip}}}(m)$, for such values of $m$, is smaller than $\Delta_{\textsf{Var}}$.
\begin{lemma}
	When $\sum_\ell m_\ell$ is even, for $m\in \{m_\ell\}_{\ell\in \mathcal{L}}$, we have that $\Delta_{\textsf{Var}}\geq \Delta_{\textsf{Var}_{\text{clip}}}$.
\end{lemma}
\begin{proof}
	The proof proceeds by a case-by-case analysis. First, observe that if $m = m^\star$, we have that $\Gamma_\ell = m_\ell$, for all $\ell \in \mathcal{L}$, implying that $\Delta_\textsf{Var} = \Delta_{\textsf{Var}_{\text{clip}}}$. Hence, in what follows, we restrict attention to the case when $m\in [m_\star:m^\star-1]$. Four possible scenarios arise:
	\begin{enumerate}
		\item $\sum_\ell m_\ell\leq 2m$: In this case, note that
		\[
		\sum_\ell \Gamma_{\ell}< \sum_\ell m_\ell\leq 2m<2m^\star.
		\]
		Hence, we have that $\Delta_\textsf{Var} = U^2/4$, with  $\Delta_{\textsf{Var}_{\text{clip}}} = U^2/4$, if $\sum_\ell \Gamma_{\ell}$ is even, and $\Delta_{\textsf{Var}_{\text{clip}}} = \frac{U^2}{4}\cdot \left(1-\frac{1}{(\sum_\ell \Gamma_{\ell})^2}\right)$, if $\sum_\ell \Gamma_{\ell}$ is odd. Clearly, the statement of the lemma is true in this case.
		\item $\sum_\ell \Gamma_{\ell}\leq 2m\leq \sum_\ell m_\ell \leq 2m^\star$: Here too, $\Delta_\textsf{Var} = U^2/4$, with  $\Delta_{\textsf{Var}_{\text{clip}}} = U^2/4$, if $\sum_\ell \Gamma_{\ell}$ is even, and $\Delta_{\textsf{Var}_{\text{clip}}} = \frac{U^2}{4}\cdot \left(1-\frac{1}{(\sum_\ell \Gamma_{\ell})^2}\right)$, if $\sum_\ell \Gamma_{\ell}$ is odd. The lemma thus holds in this case as well.
		\item $2m<\sum_\ell \Gamma_{\ell}$: In this case, observe that
		\[
		2m<\sum_\ell \Gamma_{\ell} < \sum_\ell m_\ell.
		\]
		Hence, we have that $\Delta_{\textsf{Var}_{\text{clip}}} = \frac{U^2\ m(\sum_\ell \Gamma_\ell - m)}{\left(\sum_\ell \Gamma_{\ell}\right)^2}$. First, consider the case where $\sum_\ell m_\ell>2m^\star$.
		 
		 Again, without loss of generality, assume that $m_1\geq m_2\geq \ldots \geq m_L$. Let us define $\eta(t):= \frac{U^2\ m^\star(t-m^\star)}{t^2}$, where $t\in (0,\infty)$. It is easy to show that $\frac{\d \eta}{\d t}< 0$, implying that $\eta(t)$ is decreasing in its argument $t$. Furthermore, if $m = m_1 = m^\star$, it is easy to see that $\Delta_{\textsf{Var}_{\text{clip}}} = \Delta_{\textsf{Var}}$. Now, suppose that $m = m_{j+1}$, for some $j\in [L-1]$, such that $m<m^\star = m_1$. Then, since $\sum_\ell m_\ell < (j+1)m_{1}+\sum_{\ell = j+2}^{L} m_\ell =: c$, we have by the above analysis of the function $\eta$ that
		 \[
		 \Delta_{\textsf{Var}}> \frac{U^2m^\star((j+1)m_1+c-m^\star)}{((j+1)m_1+c)^2} =: \alpha.
		 \]
		 We next show that $\alpha\geq \Delta_{\textsf{Var}_{\text{clip}}}(m)$. To this end, observe that 
		 \[
		 \Delta_{\textsf{Var}_{\text{clip}}} = \frac{U^2m((j+1)m+c-m)}{((j+1)m+c)^2}.
		 \]
		 The result follows by explicitly computing $\alpha - \Delta_{\textsf{Var}_{\text{clip}}}$ and arguing that this difference is non-negative, so long as $j\geq 2$, and hence, in particular, when $m<m^\star$. Hence, when $\frac{m_\star}{m^\star} >\frac{2}{L}$ and $m\in \{m_\ell\}$, we have that $\Delta_{\textsf{Var}_{\text{clip}}}< \Delta_{\textsf{Var}}$. For $\sum_\ell m_\ell \leq 2m^\star$, we have that $\Delta_{\textsf{Var}} = U^2/4\geq \Delta_{\textsf{Var}_{\text{clip}}}$, by a direct calculation.
		 \item $\sum_\ell \Gamma_{\ell}\leq 2m< 2m^\star\leq \sum_\ell m_\ell$: We claim that such a situation cannot arise, for the given choice of $\Gamma_{\ell}$, $\ell\in \mathcal{L}$. Indeed, observe that for $m\neq m^\star$, for $\sum_\ell \Gamma_{\ell} = \sum_\ell \min\{m,m_\ell\} \leq 2m$ to hold, we must have that for some $\ell_0\in \mathcal{L}$, the inequality $m_{\ell_0}>m$ holds, while $\sum_{\ell\neq \ell_0} m_{\ell}\leq m$. This then implies that
		 \begin{align*}
		 \sum_\ell m_\ell &= m_{\ell_0} + \sum_{\ell\neq \ell_0} m_\ell \leq m\\
		 &< m+m^\star < 2m^\star.
		 \end{align*}
	 	However, by assumption, we have that $2m^\star\leq \sum_\ell m_\ell$, leading to a contradiction.
	 	
	\end{enumerate}

Putting together all the cases concludes the proof of the lemma.
%
\end{proof}

%


%% file: varworst.tex
\label{app:varworst}
In this section, we shall prove Theorem \ref{thm:varworst}. Recall that we intend computing
\[
	E_{\textsf{Var}}:= \max_{\mathcal{D}\in \mathsf{D}}\lvert \textsf{Var}(\mathcal{D}) - \textsf{Var}_{\text{clip}}(\mathcal{D})\rvert,
\]
for fixed $\Gamma_\ell\in [0:m_\ell]$, for $\ell\in \mathcal{L}$, with the assumption that $\sum_{\ell}\Gamma_{\ell}>0$. For the case when $\Gamma_{\ell} = m_\ell$, for all $\ell\in \mathcal{L}$, it is clear that $\textsf{Var}(\mathcal{D}) = \textsf{Var}_{\text{clip}}(\mathcal{D})$ and hence that $E_{\textsf{Var}} = 0$. Hence, in what follows, we assume that there exists at least one user $\ell\in \mathcal{L}$ with $\Gamma_{\ell}< m_\ell$. Let $A:= \left\{S_\ell^{(j)}:\ \ell\in \mathcal{L},\ j\in [\Gamma_\ell]\right\}$, and define $$A^c:= \left\{S_\ell^{(j)}:\ \ell\in \mathcal{L},\ j\in [\Gamma_\ell+1:m_\ell]\right\}.$$

Now, two cases can possibly arise: (i) when {$|A|\geq |A^c|$}, and (ii) when {$|A|< |A^c|$}. Consider first case (i). Similar to the arguments made in the proof of Proposition \ref{prop:varsens}, when $\sum_\ell m_\ell$ is even, we have that
\[
E_{\textsf{Var}}\leq \max_{\mathcal{D}\in \mathsf{D}} \textsf{Var}(\mathcal{D}) = \max \textsf{Var}(X) = \frac{U^2}{4},
\]
for $X\sim \text{Unif}(A\cup A^c)$. Furthermore, equality above is attained when all samples in $A$ take the value $0$, and $\frac{|A^c|-|A|}{2}$ samples in $A^c$ take the value $0$ and the remaining samples take the value $U$. This then results in exactly $\frac{\sum_\ell m_\ell}{2}$ samples being $0$ and an equal number of samples being $U$, resulting in a variance of $U^2/4$. Further, when $\sum_\ell m_\ell$ is odd, we have that
\[
E_{\textsf{Var}}\leq \max_{\mathcal{D}\in \mathsf{D}} \textsf{Var}(\mathcal{D}) = \max \textsf{Var}(X) = \frac{U^2}{4}\cdot \left(1-\frac{1}{(\sum_\ell m_\ell)^2}\right),
\]
with equality achieved when exactly $\left \lceil \frac{\sum_\ell m_\ell}{2}\right \rceil$ samples are $0$ and $\left \lfloor \frac{\sum_\ell m_\ell}{2}\right \rfloor$ samples are $U$. This then gives rise to an exact characterization of $E_{\textsf{Var}}$ when $\sum_\ell \Gamma_\ell\leq \frac{\sum_\ell m_\ell}{2}$.

The setting of case (ii), when $|A|> |A^c|$, requires more work. However, the proof in this case is quite similar to the proof of Proposition \ref{prop:varsens}. As in Appendix \ref{app:versens}, we define
\[
{\mu}(A):= \frac{1}{\sum_\ell \Gamma_\ell}\cdot \sum_{\ell\in \mathcal{L}}\sum_{j=1}^{\Gamma_\ell} {S}_\ell^{(j)}
\]
\text{and}
\[
{\mu}({A}^c):= \frac{1}{\sum_\ell (m_\ell-\Gamma_\ell)}\cdot \sum_{\ell\in \mathcal{L}}\sum_{j=\Gamma_\ell+1}^{m_\ell} {S}_\ell^{(j)}
\]
as the sample means of the samples in the sets $A$ and $A^c$, respectively. Further, let $\mu = \mu(\mathcal{D})$. The following lemma then holds.
\begin{lemma}
	\label{lem:app3}
	When $|A|>|A^c|$, we have that
	\[
	E_\textsf{Var} = \max_{\mathcal{D}\in \mathsf{D}} \left(\textsf{Var}(\mathcal{D}) - \textsf{Var}_\text{clip}(\mathcal{D})\right).
	\]
\end{lemma}
\begin{proof}
\begin{align}
	\textsf{Var}(\mathcal{D})&= \E_{X\sim \text{Unif}(A\cup A^c)}\left[(X-\mu)^2\right] \notag\\
	&= \E\left[(X-\mu)^2\mid X\in A\right]\cdot \P(A)+\E\left[(X-\mu)^2\mid X\in A^c\right]\cdot P(A^c) \notag\\
	&= (\mu-\mu(A))^2\cdot P(A)+\E\left[(X-\mu(A))^2\mid X\in A\right]\cdot P(A)+\notag \\
	&\ \ \ \ \ \ \ \ \ \ \ (\mu-\mu(A^c))^2\cdot P(A^c)+\E\left[(X-\mu(A^c))^2\mid X\in A^c\right]\cdot P(A^c), \label{eq:app4}
\end{align}
where we abbreviate $\Pr[X\in T]$ as $P(T)$, for some set $T\subseteq A\cup A^c$. The last equality holds for reasons similar to those in \eqref{eq:app1} and \eqref{eq:app2}.

Next, note that 
\begin{align}
	\textsf{Var}_\text{clip}(\mathcal{D})&= \E_{X'\sim \text{Unif}(A)}\left[(X'-\mu(A))^2\right]\notag\\
	&= \E\left[(X-\mu(A))^2\mid X\in A\right]P(A)+  \E\left[(X-\mu(A))^2\mid X\in A\right]P(A^c). \label{eq:app5}
\end{align}

Putting together \eqref{eq:app4} and \eqref{eq:app5} and noting that, conditioned on the event $\{X\in A\}$, we have that $X$ is uniform on the values in the set $A$, we get that
\begin{align}
	\lvert \textsf{Var}(\mathcal{D}) &- \textsf{Var}_{\text{clip}}(\mathcal{D})\rvert = \notag\\
	& \bigg \vert (\mu-\mu(A))^2 P(A) + (\mu-\mu(A^c))^2 P(A^c)\ + \notag\\
	&\E\left[(X -\mu(A^c))^2\mid X\in A^c\right] P(A^c) - E\left[(X-\mu(A))^2\mid X\in A\right]P(A^c)\bigg \vert. \label{eq:app6}
\end{align} 
Now, consider a dataset $\overline{\mathcal{D}}$ such that the samples in $A$ take the value $0$ and the samples in $A^c$ take the value $U$. Clearly, we have that
\begin{align}
	E_\textsf{Var} &\geq E_\textsf{Var}(\overline{\mathcal{D}}) = \frac{U^2\cdot |A|\cdot |A^c|}{\left(\sum_\ell m_\ell\right)^2}. \label{eq:app7}
\end{align}
Furthermore, observe that 
\begin{align}
	E_\textsf{Var} 
	&= \max\left\{\max_\mathcal{D} \left(\textsf{Var}(\mathcal{D}) - \textsf{Var}_\text{clip}(\mathcal{D})\right),  \max_\mathcal{D} \left(\textsf{Var}_\text{clip}(\mathcal{D}) - \textsf{Var}(\mathcal{D})\right)\right\}. \label{eq:app9}
\end{align}
Now, from \eqref{eq:app6}, note that
\begin{align}
\max_\mathcal{D} \left(\textsf{Var}_\text{clip}(\mathcal{D}) - \textsf{Var}(\mathcal{D})\right) &\leq \E\left[(X-\mu(A))^2\mid X\in A\right]P(A^c) \notag\\
&\leq \frac{U^2|A^c|}{4\cdot \sum_\ell m_\ell}. \label{eq:app8}
\end{align}
By comparing \eqref{eq:app7} and \eqref{eq:app8},  plugging back into \eqref{eq:app9}, and noting that $|A|>|A^c|$, we obtain the statement of the lemma.
\end{proof}
Thus, from the above lemma and from \eqref{eq:app6}, we obtain that when $|A|>|A^c|$, 
\begin{align}
	E_\textsf{Var} = 
	 (\mu-\mu(A))^2& P(A) + (\mu-\mu(A^c))^2 P(A^c)\ + \notag\\
	&\E\left[(X -\mu(A^c))^2\mid X\in A^c\right] P(A^c) - E\left[(X-\mu(A))^2\mid X\in A\right]P(A^c). \notag
\end{align} 
Now, clearly, we have that $E_\textsf{Var}$ above is maximized by setting $X = \mu(A)$, when $X\in A$, or, in other words, setting $S_\ell^{(j)} = \mu(A)$, for all $\ell\in \mathcal{L}$ and $j\in [\Gamma_\ell]$. We thus obtain the following lemma:
\begin{lemma}
	\label{lem:app4}
	When $|A|>|A^c|$, we have that
	\[
	E_\textsf{Var} = \max_{\mathcal{D}:\ S_\ell^{(j)} = \mu(A), \forall S_\ell^{(j)}\in A} \textsf{Var}(\mathcal{D}).
	\]
\end{lemma}
Note the similarity between Lemma \ref{lem:app4} and Lemma \ref{lem:app2} in Appendix \ref{app:versens}. The proof of Theorem \ref{thm:varworst} is then immediate.
\begin{proof}[Proof of Theorem \ref{thm:varworst}]
	Following on from Lemma \ref{lem:app4}, by arguments analogous to those in the proof of Proposition \ref{prop:varsens} in Appendix \ref{app:versens}, we get that when $|A|>|A^c|$, $E_\textsf{Var} = \frac{U^2\cdot |A|\cdot |A^c|}{\left(\sum_\ell m_\ell\right)^2}$, which in turn equals $\frac{U^2\cdot \sum_\ell \Gamma_\ell\cdot \sum_{\ell'} (m_{\ell'} - \Gamma_{\ell'})}{\left(\sum_\ell m_\ell\right)^2}$. The case when $|A|\leq |A^c|$ was already discussed earlier, wherein $E_\textsf{Var} = \frac{U^2}{4}$, if $\sum_\ell m_\ell$ is even, and $E_\textsf{Var} = \frac{U^2}{4}\cdot \left(1-\frac{1}{(\sum_\ell m_\ell)^2}\right)$, if $\sum_\ell m_\ell$ is odd.
\end{proof}